\RequirePackage{currfile}
\documentclass[runningheads,a4paper,english]{llncs}[2022/01/12]

\newif\ifmain%
\ifcurrfilename{main.tex}{\maintrue}{\mainfalse}
\ifmain
\usepackage[appendix=append,bibliography=common]{apxproof}
\else
\usepackage[appendix=inline,bibliography=common]{apxproof}
\fi

\newtheoremrep{theorem}{Theorem}
\newtheoremrep{proposition}[theorem]{Proposition}
\newtheoremrep{lemma}[theorem]{Lemma}

\ifmain
\usepackage{upquote}
\usepackage[ngerman,main=english]{babel}
\addto\extrasenglish{\languageshorthands{ngerman}\useshorthands{"}}
\usepackage{regexpatch}
\makeatletter
\edef\switcht@albion{%
  \relax\unexpanded\expandafter{\switcht@albion}%
}
\xpatchcmd*{\switcht@albion}{ \def}{\def}{}{}
\xpatchcmd{\switcht@albion}{\relax}{}{}{}
\edef\switcht@deutsch{%
  \relax\unexpanded\expandafter{\switcht@deutsch}%
}
\xpatchcmd*{\switcht@deutsch}{ \def}{\def}{}{}
\xpatchcmd{\switcht@deutsch}{\relax}{}{}{}
\edef\switcht@francais{%
  \relax\unexpanded\expandafter{\switcht@francais}%
}
\xpatchcmd*{\switcht@francais}{ \def}{\def}{}{}
\xpatchcmd{\switcht@francais}{\relax}{}{}{}
\makeatother
\fi

\ifmain
\usepackage[%
    rm={oldstyle=false,proportional=true},%
    sf={oldstyle=false,proportional=true},%
    tt={oldstyle=false,proportional=true,variable=false},%
    qt=false%
]{cfr-lm}
\usepackage[T1]{fontenc}
\fi

\usepackage[hyphens]{url}
\makeatletter
\g@addto@macro{\UrlBreaks}{\UrlOrds}
\makeatother
\usepackage{color}

\usepackage{hyperref}
\hypersetup{
  hidelinks,
  colorlinks=true,
  allcolors=black,
  pdfstartview=Fit,
  breaklinks=true
}

\usepackage[all]{hypcap}

\ifmain
\usepackage[
  babel=true,
  expansion=alltext,
  protrusion=alltext-nott,
  final
]{microtype}
\DisableLigatures{encoding = T1, family = tt* }
\fi

\ifmain
\usepackage[
  all=normal,
  paragraphs,
  floats,
  charwidths,
  indent,
]{savetrees}
\fi

\usepackage[backend=biber,style=lncs]{biblatex}
\ifmain
\SetExpansion
[ context = sloppy,
  stretch = 30,
  shrink = 60,
  step = 5 ]
{ encoding = {OT1,T1,TS1} }
{ }
\fi

\usepackage{orcidlink}
\newcommand{\orc}[1]{\ifmain\orcidID{#1}\else\orcidlink{#1}\fi}

\usepackage{qmpst-packages}
\usepackage{qmpst-macros}

\def\finex{{\unskip\nobreak\hfil
            \penalty50\hskip1em\null\nobreak\hfil$\diamond$
            \parfillskip=0pt\finalhyphendemerits=0\endgraf}}

\newcommand{\named}[2]{#1 \triangleright #2}

\newcommand{\measure}[2]{\expMeasure{#1}{#2}}
\newcommand{\measureOp}{\expFmt{meas}}
\newcommand{\para}[2]{#1\pp #2}
\newcommand{\pp}{\mid}
\newcommand{\ifff}{{\bf if\ }}
\newcommand{\then}{{\bf then\ }}
\newcommand{\els}{{\bf else\ }}
\newcommand{\tl}[1]{{\bf tl}(#1)}
\newcommand{\hd}[1]{{\bf hd}(#1)}

\newcommand{\fst}[1]{{\bf fst}(#1)}
\newcommand{\snd}[1]{{\bf snd}(#1)}
\newcommand{\gayIn}[2]{#1 ? #2}
\newcommand{\gayOutput}[1]{{\bf{output}}\,#1}

\newcommand{\trans}[3]{#1 \rightarrow #2: #3}

\newcommand{\plts}[2]{
    \ifempty{#2}{%
        \xlongrightarrow[]{#1}_{1}
    }{%
        \xlongrightarrow[]{#1}_{#2}
    }
}
\newcommand{\lts}[1]{\xlongrightarrow[]{#1}}

\newcommand{\Lts}[1]{\xRightarrow[]{#1}}

\newcommand{\regone}{\tau}
\newcommand{\regtwo}{\tau'}
\newcommand{\sysone}{\sysM}
\newcommand{\systwo}{\sysMi}

\newcommand{\prob}{\pi}

\newcommand{\qregone}{\mathsf{Q}}
\newcommand{\qregtwo}{\mathsf{Q}'}

\newcommand{\pair}[2]{\langle #1,#2\rangle}

\newcommand{\QS}[1]{\expUnitary[QS]{#1}}
\newcommand{\qow}[1]{\expUnitary[qow]{#1}}
\newcommand{\qowOp}{\expUnitaryOp[qow]}

\newcommand{\labels}{\mathsf{L}}
\newcommand{\varSet}{\mathcal{V}}

\newcommand{\hilb}[1]{\mathcal{H}[#1]}
\newcommand{\concconf}{\pair{\qbregone}{\sysone}}

\newcommand{\basis}[1]{\mathcal{B}[#1]}
\newcommand{\qprob}{\mathbb{P}}
\newcommand{\qproj}{\hat{\mathsf{P}}}
\newcommand{\qnew}{\mathsf{N}}
\newcommand{\rint}{[0,1]}
\newcommand{\cmplx}{\mathbb{C}}
\newcommand{\qbregone}{\mathscr{Q}}
\newcommand{\qbregtwo}{\mathscr{Q}'}
\newcommand{\setsub}[2]{{#1}{\setminus}\{#2\}}
\newcommand{\setadd}[2]{{#1}{\cup}\{#2\}}
\newcommand{\sets}[1]{\{#1\}}

\begin{document}
\title{
  Towards Quantum Multiparty Session Types
  \thanks{Partially supported by the EU Marie Sk\l{}odowska-Curie action
    ReGraDe-CS (No: 101106046), by French ANR project SmartCloud
    ANR-23-CE25-0012, by INdAM – GNCS 2024 project MARVEL, code CUP
    E53C23001670001, and by the Italian Ministry of University and Research under
    PNRR - M4C2 - I1.4 Project CN00000013 ``National Centre for HPC, Big Data and
    Quantum Computing''}
}
\titlerunning{Towards QMPSTs}
\author{
  {Ivan Lanese}\inst{1}\orc{0000-0003-2527-9995}
  \and
  {Ugo Dal Lago}\inst{1}\orc{0000-0001-9200-070X}
  \and
  {Vikraman Choudhury}\inst{1}\orc{0000-0003-2030-8056}
}
\authorrunning{Lanese et al.}
\institute{OLAS Team, University of Bologna \& INRIA}
\maketitle
\begin{abstract}
  Multiparty Session Types (MPSTs) offer a structured way of specifying
communication protocols and guarantee relevant communication
properties, such as deadlock-freedom.
In this paper, we extend a minimal MPST system with quantum data and operations,
enabling the specification of quantum protocols.
Quantum MPSTs (QMPSTs) provide a formal notation to describe quantum protocols,
both at the abstract level of global types,
describing which communications can take place in the system and their dependencies,
and at the concrete level of local types and quantum processes,
describing the expected behavior of each participant in the protocol.
Type-checking relates these two levels formally,
ensuring that processes behave as prescribed by the global type.
Beyond usual communication properties,
QMPSTs also allow us to prove that qubits are owned by a single process at any time,
capturing the quantum no-cloning and no-deleting theorems.
We use our approach to verify four quantum protocols from the literature,
respectively Teleportation, Secret Sharing, Bit-Commitment, and Key Distribution.

\end{abstract}

\setlength{\emergencystretch}{3em}

\section{Introduction}
\label{sec:intro}
Quantum protocols involve the exchange of (quantum) information between multiple parties in quantum networks,
giving rise to complex interaction patterns, interleaved with manipulations of quantum states.
This raises the need for tools and techniques to specify, analyse, and verify such protocols.
However, while many languages, notations, and formal calculi exist to describe quantum circuits,
fewer options exist for formally describing quantum multiparty protocols
(see the discussion in the recent systematic survey~\cite{qplSurvey}).
In fact, there does not exist a mainstream formal way to describe quantum protocols. This is witnessed by the fact that
the Quantum Protocol Zoo~\cite{qpzoo},
a well-known library of quantum protocols,
relies on natural language -- hence ambiguous -- descriptions, paired with Python implementations.

Existing formalisms for quantum protocols include imperative languages such as
LanQ~\cite{LanQ} and QMCLANG~\cite{Papanikolaou09,DavidsonGMNP12},
and process calculi such as
CQP~\cite{GayN05}, $\textrm{CCS}^\textrm{q}$~\cite{gay18} and lqCCS~\cite{gadducciQBisim}.
The former do not allow for an abstract representation of the protocol,
while the latter are not easy to understand (see the discussion in~\cref{sec:concl} for details).
The analysis in~\cite{GayN05} on the shortcomings of their CQP approach reports:
``The proliferation of channels is a consequence of the fact
that our type system associates a unique type with each channel.
Introducing session types would allow a single channel to be used for the entire protocol''.

Following this hint in~\cite{GayN05}, we propose a quantum extension to Multiparty Session Types
(MPSTs)~\cite{HondaYC16,Huttel+16},
dubbed Quantum MPSTs (QMPSTs),
as a formal session-typed language to describe quantum protocols.
QMPSTs provide both an abstract view -- the \emph{global type},
and a concrete view -- a multiparty system made of named \emph{quantum processes}.
The global type describes the expected pattern of interactions from a global viewpoint,
clarifying which communications take place, between which participants, and in which order.
QMPST processes describe the sequence of actions taken by each participant inside the protocol,
including quantum actions such as the application of unitary transformations or the measurement of qubits.
The two views are formally related --
\emph{type checking} ensures that processes actually behave as prescribed by the global type (\cref{thm:subject-reduction}).
The framework of MPSTs also ensures relevant communication properties by construction,
such as progress (\cref{thm:progress}).
Further, in our approach, messages will directly target participants,
denoted by their name,
avoiding the proliferation of channels (as in CQP),
and clarifying the expected sender and receiver of each message (cf.~\cref{ex:secret}).


\vspace{-0.5\baselineskip}
\paragraph*{Outline and Contributions}

The main contribution of this paper is a novel framework for specifying quantum protocols using MPSTs,
building on classical approaches, but extending them with quantum features.
There is no other extension of MPSTs for quantum protocols in the literature, to the best of our knowledge.
\begin{itemize}[leftmargin=*]
      \item In~\cref{sec:type}, we introduce our calculus for QMPSTs.
            At the level of types, the main novelty is that data include references to qubits,
            and such references need to be treated in a linear way (cf.~\cref{thm:unique})
            due to the no-cloning and no-deleting theorems~\cite{noclone} of quantum computing.
            At the level of processes, as mentioned before,
            we introduce quantum operations such as unitary transformations and measurement.
      \item In~\cref{sec:semantics}, we provide a concrete operational semantics for QMPSTs.
            We define a labelled transition system for quantum multiparty systems,
            that is both non-deterministic and probabilistic,
            to account for multiparty system communication,
            and the probabilistic nature of quantum measurements.
      \item In~\cref{sec:results}, we establish meta-theoretical results for our calculus.
            We prove the usual communication safety properties of MPSTs, namely,
            subject reduction (\cref{thm:subject-reduction}), session fidelity (\cref{thm:session-fidelity}),
            progress (\cref{thm:progress}), and type safety (\cref{thm:safety}),
            and quantum safety properties, namely,
            unique ownership of qubits (\cref{thm:unique}) and qubit safety (\cref{thm:qubit-safety}).
      \item In~\cref{sec:examples}, we show the expressive power and clarity of our approach
            by applying it to four quantum protocols from the literature,
            namely Teleportation (used as a running example), Secret Sharing (\cref{ex:secret}, in~\cref{sec:examples}),
            Bit-Commitment (\cref{ex:qbc}, in~\cref{sec:examples}), and Key Distribution
            \ifmain
            (\cref{ex:key}, deferred to the appendix for space reasons).
            \else
            (\cref{ex:key}).
            \fi
\end{itemize}
We remark that some familiarity with quantum computing is needed to understand the paper --
for standard notions, we refer to~\cite{nielsen2010quantum}.
\ifmain
      Details of proofs and auxiliary results can be found in the supplementary appendices.
\fi

\section{Type system}
\label{sec:type}

We start from a minimal MPSTs paradigm (such as the one in~\cite{gentleIntro}) to avoid cluttering the work with
technicalities that are not related to the interplay between MPSTs and quantum computing and which would spoil the
cleanliness of the language.
This approach is also useful in understanding which advanced features of MPSTs are useful in the quantum setting,
and which ones are unnecessary.

We follow the usual top-down approach to MPSTs, whereby we first give a global type describing the interactions in the
system from a global point of view. We then project the global type into local types, describing the required behavior
of each participant.
Local types are used to type check the processes describing participants.
The main novelty of our type system is that basic types include qubits ($\tyQBit$),
and qubit variables are treated linearly.

Amongst the common features in approaches to MPSTs (such as~\cite{ScalasY19,HondaYC16}),
we drop support for multiple sessions and delegation, and only use a single channel.
These features are not useful for our case studies in~\cref{sec:examples},
nor in other examples of quantum protocols we have looked at
(see, e.g., the library of quantum protocols mentioned above~\cite{qpzoo}).
The most important feature of MPSTs that we use is the notion of internal and external choice
(or send and receive, respectively), and selection and branching constructs in processes,
which allows processes to exchange both classical and quantum data.
Hence, following Selinger's famous slogan \emph{``quantum data + classical control''}
for quantum lambda-calculus~\cite{SelingerV06},
the slogan for our quantum process calculus is \emph{``quantum data + classical control + classical choice''}.

\begin{figure*}[t]%
  \centerline{\(
    \begin{array}{r@{\quad}c@{\quad}l@{\quad}l}
      \gtG      & \bnfdef                                                    &
      \gtComm{\roleP}{\roleQ}{i \in I}{\gtLab[i]}{\tyGround[i]}{\gtG[i]}
                &
      {\footnotesize\text{Transmission or Interaction}}
      \\[1em]
                & \bnfsep                                                    & \gtRec{\gtRecVar}{\gtG} \quad \bnfsep \quad \gtRecVar
      \quad \bnfsep \quad \gtEnd
                &
      \text{\footnotesize Recursion, Type variable, Termination}
      \\[1em]
      \stT
                & \bnfdef                                                    & \stExtSum{\roleP}{i \in I}{\stChoice{\stLab[i]}{\tyGround[i]} \stSeq \stT[i]}
                & \text{\footnotesize External choice (receive)}
      \\[1em]
                & \bnfsep                                                    & \stIntSum{\roleP}{i \in I}{\stChoice{\stLab[i]}{\tyGround[i]} \stSeq \stT[i]}
                & \text{\footnotesize Internal choice (send)}
      \\[1em]
                & \bnfsep                                                    & \stRec{\stRecVar}{\stT} \quad \bnfsep \quad \stRecVar      \quad \bnfsep \quad \stEnd
                & \text{\footnotesize Recursion, Type variable, Termination}
      \\[1em]
      \tyGround & \bnfdef                                                    & \tyBit \bnfsep \tyQBit \bnfsep \tyUnit \bnfsep \ldots
                & \text{\footnotesize Basic types}
    \end{array}
    \)}
  \caption{Syntax of types.}
  \label{fig:syntax-global-type}%
  \label{fig:syntax-local-type}%
  \label{fig:syntax-mpst}
\end{figure*}

We first define the syntax of types.
In the following definition,
$\tuple{q_i}{i \in I}$ and $\set{q_i}_{i \in I}$ denote indexed tuples and indexed sets, respectively.
We also denote tuples as $\tq$, leaving the indexing implicit.
We also write $\widetilde{\stEnvMap{x}{\stS}}$ to mean a tuple $\tx$ of typed variables, each of type $\stS$.

\begin{definition}[Syntax of types]
  The syntax of \emph{global types}, \emph{local types}, and \emph{basic types},
  ranged over respectively by $\gtG$, $\stT$ and $\tyGround$, is defined in~\cref{fig:syntax-global-type},
  where $\roleP, \roleQ$ are \emph{participants} (or roles) and $\stLab[i]$ denote \emph{labels}.
  Labels in a global type are assumed to be disjoint, and label sets are non-empty.
  Basic types include at least $\tyBit$s, $\tyUnit$, integers, and \emph{references to qubits}.
\end{definition}
In choices we may drop braces when the index set is a singleton.
References to qubits are needed since, due to entanglement, allowing
values of distinct qubits to be related (indeed measuring one of them may impact the value of the other), it is not possible to
provide separate descriptions of single qubits. Following~\cite{SelingerV06,dal2015applicative}, we hence rely on a
global qubit register, abstracting away from the actual physical location of
the qubits, to describe the whole quantum state of a system. References
are pointers (or names) to quantum states inside this register.

Labels are used as usual to communicate which branch is taken in a
choice. We drop labels when the choice has a single branch, and the
type of it, if it is $\tyUnit$.
Global recursive types are contractive, that is, each occurence of the recursive type variable inside the body of the
recursive definition is guarded.
We also assume the equi-recursive equation: $\gtRec{\gtRecVar}{\gtG} =
  \gtG\subst{\gtRecVar}{\gtRec{\gtRecVar}{\gtG}}$ holds.

\begin{example}[Quantum teleportation {\cite[Fig.~14]{gay18}} - Description and global type]\label{ex:teleportationGlobal}
  The quantum teleportation protocol~\cite{bennett1993teleporting}
  allows one to transmit a quantum state via a \emph{non-quantum} medium.
  A slightly different version of the same is
  reported in~\cite[Fig.~5]{GayN05} under the name of ``Quantum
  teleportation with EPR source''.

  In order to more easily model the example we extend our
  data types with $n$-uples of bits, denoted $\tyBit^n$.
  We can then write the global
  type below. We use a dedicated participant called
  $\roleEnv$, for environment, where other participants can take
  input from and send output to. This avoids the need for specific
  input and output primitives as in~\cite{gay18}.
  \[
    \begin{array}{l}
      \gtG = \trans{\roleEnv}{\roleAlice}{(\tyQBit)} \gtSeq
      \trans{\roleSource}{\roleAlice}{(\tyQBit)} \gtSeq
      \trans{\roleSource}{\roleBob}{(\tyQBit)} \gtSeq \\
      \qquad \qquad \trans{\roleAlice}{\roleBob}{(\tyBit^2)} \gtSeq
      \gtCommSingle{\roleBob}{\roleEnv}{}{\tyQBit}{
        \gtEnd
      }
    \end{array}
  \]
  The idea is that $\roleAlice$ takes a $\tyQBit$ from
  $\roleEnv$, and wants to send it to $\roleBob$ via a classical
  communication channel (indeed
  $\roleAlice$ only sends to $\roleBob$ two classical $\tyBit$s).
  $\roleBob$ can then give the $\tyQBit$ to $\roleEnv$.
  In order to do so, $\roleAlice$ and $\roleBob$ take from some
  $\roleSource$ two entangled qubits.\finex
\end{example}

The set of roles or participants in a global type $\gtG$ is given by $\gtRoles{\gtG}$:
\[
  \begin{array}{rcl}
    \gtRoles{\gtComm{\roleP}{\roleQ}{i \in I}{\gtLab[i]}{\tyGround[i]}{\gtG[i]}} & = & \set{\roleP, \roleQ} \cup \bigcup_{i \in I} \gtRoles{\gtG[i]} \\
    \gtRoles{\gtRec{\gtRecVar}{\gtG}}                                            & = & \gtRoles{\gtG}                                                \\
    \gtRoles{\gtRecVar} = \gtRoles{\gtEnd}                                       & = & \emptyset
  \end{array}
\]
Projection allows one to derive which local type each participant needs to
follow in order for their composition to behave as prescribed by the
global type. We take the notion of projection below from~\cite{ScalasY19}.
\begin{definition}[Global Type Projection]%
  \label{def:global-proj}%
  \label{def:local-type-merge}%
  The \emph{projection of a global type $\gtG$ onto a role $\roleP$}, %
  written $\gtProj{\gtG}{\roleP}$, %
  is:

  \smallskip
  \centerline{\(%
    \begin{array}{c}
      \gtProj{\left(%
        \gtCommSmall{\roleQ}{\roleR}
        {i \in I}{\gtLab[i]}{\stS[i]}{\gtG[i]}%
        \right)}{\roleP}%
      =\!%
      \left\{%
      \begin{array}{@{}l@{\hskip 5mm}l@{}}
        \stIntSum{\roleR}{i \in I}{ %
          \stChoice{\stLab[i]}{\stS[i]} \stSeq (\gtProj{\gtG[i]}{\roleP})%
        }%
         & \text{\footnotesize %
          if\, $\roleP = \roleQ$%
        }%
        \\[2mm]%
        \stExtSum{\roleQ}{i \in I}{%
          \stChoice{\stLab[i]}{\stS[i]} \stSeq (\gtProj{\gtG[i]}{\roleP})%
        }%
         & \text{\footnotesize %
          if\, $\roleP = \roleR$
        }
        \\[2mm]%
        \stMerge{i \in I}{\gtProj{\gtG[i]}{\roleP}}%
         &                     %
        \text{\footnotesize%
          if\, $\roleP \neq \roleQ$ \,and \,$\roleP \neq \roleR$
        }%
      \end{array}
      \right.
      \\[8mm]%
      \gtProj{(\gtRec{\gtRecVar}{\gtG})}{\roleP}%
      \;=\;%
      \left\{%
      \begin{array}{@{\hskip 0.5mm}l@{\hskip 3mm}l@{}}
        \stRec{\stRecVar}{(\gtProj{\gtG}{\roleP})}%
         & %
        \text{\footnotesize%
          if\,
          $\roleP \in \gtRoles{\gtG}$ \,or\,
          $\fv{\gtRec{\gtRecVar}{\gtG}} \neq \emptyset$
        }%
        \\%
        \stEnd%
         & %
        \text{\footnotesize%
          otherwise}
      \end{array}
      \right.%
      \quad\qquad%
      \begin{array}{@{}r@{\hskip 1mm}c@{\hskip 1mm}l@{}}
        \gtProj{\gtRecVar}{\roleP}%
         & = & %
        \stRecVar%
        \\%
        \gtProj{\gtEnd}{\roleP}%
         & = & %
        \stEnd%
      \end{array}
    \end{array}
    \)}%

  \smallskip%
  \noindent%
  where
  $\stMerge{}{}$ is the full-merging operation, defined as:

  \smallskip%
  \centerline{\(%
    \begin{array}{c}%
      \textstyle%
      \stIntSum{\roleP}{i \in I}{\stChoice{\stLab[i]}{\stS[i]} \stSeq \stT[i]}%
      \,\stBinMerge\,%
      \stIntSum{\roleP}{i \in I}{\stChoice{\stLab[i]}{\stS[i]} \stSeq \stTi[i]}%
      \;=\;%
      \stIntSum{\roleP}{i \in I}{\stChoice{\stLab[i]}{\stS[i]} \stSeq (\stT[i]
        \stBinMerge \stTi[i])}%
      \\[1mm]%
      \stRec{\stRecVar}{\stT} \,\stBinMerge\, \stRec{\stRecVar}{\stTi}%
      \,=\,%
      \stRec{\stRecVar}{(\stT \stBinMerge \stTi)}%
      \quad%
      \stRecVar \,\stBinMerge\, \stRecVar%
      \,=\,%
      \stRecVar%
      \quad%
      \stEnd \,\stBinMerge\, \stEnd%
      \,=\,%
      \stEnd%
      \\[1mm]
      \stExtSum{\roleP}{i \in I}{\stChoice{\stLab[i]}{\stS[i]} \stSeq \stT[i]}%
      \,\stBinMerge\,%
      \stExtSum{\roleP}{\!j \in J}{\stChoice{\stLab[j]}{\stS[j]} \stSeq \stTi[j]}%
      \;=\;  \stExtSum{\roleP}{\!k \in I \cup J}{\stChoice{\stLab[k]}{\stS[k]} \stSeq \stTii[k]}
      \\[1mm]
      \text{where} \quad %
      \stTii[k]
      \;=\;%
      \left\{%
      \begin{array}{@{\hskip 0.5mm}l@{\hskip 3mm}l@{}}
        \stT[k]  \stBinMerge \stTi[k]
         & %
        \text{\footnotesize%
          if\,
          $\stFmt{k} \in \stFmt{I \cap J}$
        }%
        \\%
        \stT[k]
         & %
        \text{\footnotesize%
          if\,
          $\stFmt{k} \in \stFmt{I \setminus J}$
        }%
        \\%
        \stTi[k]
         & %
        \text{\footnotesize%
          if\,
          $\stFmt{k} \in \stFmt{J \setminus I}$
        }
      \end{array}
      \right.%
    \end{array}
    \)}%
\end{definition}
\noindent
Intuitively, participants involved in a choice act as prescribed by
the choice. Participants which are not involved need either to behave
the same in the two branches, or to receive messages with distinct
labels telling them how to behave.

\begin{example}[Quantum teleportation - Local types]
  We can obtain the local types by projecting the global type in~\cref{ex:teleportationGlobal} on the different participants:
  \begin{eqnarray*}
    \gtProj{\gtG}{\roleFmt{Alice}} &=& \stExtSumSing{\roleFmt{Env}}{}{\stChoice{}{\tyQBit} \stSeq \stExtSumSing{\roleFmt{Source}}{}{\stChoice{}{\tyQBit} \stSeq \stIntSumSing{\roleFmt{Bob}}{}{\stChoice{}{\tyBit^2} \stSeq \stEnd}}}\\
    \gtProj{\gtG}{\roleFmt{Bob}} &=& \stExtSumSing{\roleFmt{Source}}{}{\stChoice{}{\tyQBit} \stSeq \stExtSumSing{\roleFmt{Alice}}{}{\stChoice{}{\tyBit^2} \stSeq \stIntSumSing{\roleFmt{Env}}{}{\stChoice{}{\tyQBit} \stSeq \stEnd}}}\\
    \gtProj{\gtG}{\roleFmt{Source}} &=& \stIntSumSing{\roleFmt{Alice}}{}{\stChoice{}{\tyQBit} \stSeq \stIntSumSing{\roleFmt{Bob}}{}{\stChoice{}{\tyQBit} \stSeq \stEnd}}\\
    \gtProj{\gtG}{\roleFmt{Env}} &=& \stIntSumSing{\roleFmt{Alice}}{}{\stChoice{}{\tyQBit} \stSeq \stExtSumSing{\roleFmt{Bob}}{}{\stChoice{}{\tyQBit} \stSeq \stEnd}} \hspace{4.3cm} \diamond  \end{eqnarray*}
\end{example}

A global type is well-formed if it is contractive, and can be projected onto all its roles.
We implicitly assume that all global types are well-formed.

\begin{figure*}[t]
  \centerline{\(
    \begin{array}{c}
      \inference[\iruleGtMoveComm]{
      j \in I
      }{
      \gtCommSmall{\roleP}{\roleQ}{i \in I}{\gtLab[i]}{\tyGround[i]}{\gtGi[i]}
      \,\gtMove[
      \gtCommLab{\roleP}{\roleQ}{}:{{\gtLab[j]}({\tyGround[j]})}
      ]\,
      \gtGi[j]
      }
      \\[1em]
      \inference[\iruleGtMoveCtx]{
        \forall i \in I
        \
        \gtGi[i]
        \,\gtMove[\stEnvAnnotGenericSym]\,
        \gtGii[i]
       &
        \ltsRoles{\stEnvAnnotGenericSym} \cap \setenum{\roleP, \roleQ} =
        \emptyset
      }{
        \gtCommSmall{\roleP}{\roleQ}{i \in I}{\gtLab[i]}{\stS[i]}{\gtGi[i]}
        \,\gtMove[\stEnvAnnotGenericSym]\,
        \gtCommSmall{\roleP}{\roleQ}{i \in I}{\gtLab[i]}{\stS[i]}{\gtGii[i]}
      }
    \end{array}
    \)}
  \caption{Global type transitions}
  \label{fig:gtype:red-rules}
\end{figure*}

\begin{definition}[Global Type Transitions]
  \label{def:gtype:lts-gt}
  The \emph{global type transition relation} %
  $\gtMove[\stEnvAnnotGenericSym]$ is defined in~\cref{fig:gtype:red-rules}.
  We write
  $\gtG \,\gtMove\, \gtGi$
  if there exists $\stEnvAnnotGenericSym$ such that
  $\gtG \,\gtMove[\stEnvAnnotGenericSym]\,\gtGi$.
  We also write $\gtG \,\gtMove$
  if there exists $\gtGi$ such that
  $\gtG \,\gtMove\, \gtGi$;
  $\gtMoveStar$ denotes the reflexive and transitive closure of $\gtMove$.
\end{definition}
Rule \iruleGtMoveComm~in~\cref{fig:gtype:red-rules} just allows one to execute
the first interaction. Rule \iruleGtMoveCtx~allows one to execute an
interaction $\stEnvAnnotGenericSym$ available in all the continuations provided that it has
disjoint participants from the first one.
Sets of local types projected from a global type can be used to type programs.
We extend the type system in~\cite[Fig.~12]{gay18}, which
however requires channels to carry uniform values. Like in~\cite{GayN05},
qubits are treated linearly.

The syntax of systems, processes and expressions is
in \cref{fig:syntax-processes}. We take inspiration from the syntax
in \cite[Fig.~11]{gay18}, with a few changes to fit MPSTs.
First, for us a system is a parallel composition of
processes with roles, and, send and receive actions explicitly refer to the target
role $\roleP$. This makes the protocol clearer than using channels, and helps
the technical development as well. Then, as customary for session
types, we have input choice to allow a participant to wait for several
input messages, and accept the first one that arrives. We use labels to
clarify which branch of a choice is selected.  Finally, we avoid the
input/output constructs in~\cite[Fig.~11]{gay18} (we
model them as interactions with some $\roleEnv$ironment role, as done in our running example).

We assume a set $\varSet$ for variables, which may contain either
classical values, or references to the qubit register which contains
actual quantum information. We assume that operators in expressions
only work on classical data, operations on $\tyQBit$s are explicitly
modeled as unitary operations. We assume to have a library of unitary
operations (e.g., controlled not, Hadamard, ...), ranged over by $\expUnitaryOp$. Process $\mpNil[\tq]$ is labelled with the tuple $\tq$ of $\tyQBit$s it owns. This is needed to ensure linearity of $\tyQBit$s, and to clarify which participant owns which $\tyQBit$. This is relevant since due to entanglment (s)he (in a follow up protocol) may influence other $\tyQBit$s entangled with its own.
We write just $\mpNil$ when $\tq$ is empty.





\begin{figure*}
  \centerline{\(
    \begin{array}{r@{\hskip 2mm}c@{\hskip 2mm}l@{\hskip 5mm}l}
      \expE      & \bnfdef & x                                                       & \mbox{\footnotesize(variable)}                                                     \\[0.5em]
                 & \bnfsep & q                                                       & \mbox{\footnotesize(reference inside the qubit register)}                                                    \\[0.5em]
                 & \bnfsep & c                                                       & \mbox{\footnotesize(constant)}                                                     \\[0.5em]
                 & \bnfsep & \expE \expOp \expE                                      & \mbox{\footnotesize(operation)}                                                    \\[0.5em]
      \expQBit   & \bnfdef & x                                                       & \mbox{\footnotesize(variable)}                                                     \\[0.5em]
                 & \bnfsep & q                                                       & \mbox{\footnotesize(reference inside the qubit register)}                                                    \\[0.5em]
      \mpP, \mpQ & \bnfdef & \mpNil[\tz]                                             & \mbox{\footnotesize(inaction, owning qubits in $\tz$)}                             \\[0.5em]
                 & \bnfsep & \mpSel{}{\roleP}{\mpLab}{\mpExpr}{\mpP}                 & \mbox{\footnotesize(selection/send towards role $\roleP$)}                         \\[0.5em]
                 & \bnfsep & \mpBranch{}{\roleP}{i \in I}{\mpLab[i]}{x_i}{\mpP[i]}{} & \mbox{\footnotesize(branching/receive from role $\roleP$ with $I \neq \emptyset$)} \\[0.5em]
                 & \bnfsep & \mpDef{\mpX}{\tx}{\mpP}{\mpQ}                           & \mbox{\footnotesize(process definition)}                                           \\[0.5em]
                 & \bnfsep & \mpCallSmall{\mpX}{\widetilde{\mpExpr}}                 & \mbox{\footnotesize(process call)}                                                 \\[0.5em]
                 & \bnfsep & \mpIf{e}{\mpP}{\mpQ}                                    & \mbox{\footnotesize(conditional)}                                                  \\[0.5em]
                 & \bnfsep & \expMeasure{x}{z} \mpSeq \mpP                           & \mbox{\footnotesize(measurement)}                                                  \\[0.5em]
                 & \bnfsep & \expNewQBit[x] \mpSeq \mpP                              & \mbox{\footnotesize(new qubit generation)}                                         \\[0.5em]
                 & \bnfsep & \expUnitary[]{\tz} \mpSeq \mpP                          & \mbox{\footnotesize(unitary operation)}                                            \\[0.5em]
      \sysM      & \bnfdef & {\prod\limits_{i \in I} \mpProc{\roleP[i]}{\mpP[i]}}    & \mbox{\footnotesize(multiparty system)}
    \end{array}
    \)}
  \caption{Syntax of expressions, processes, and systems}
  \label{fig:syntax-processes}
\end{figure*}

Note that $\prod_{i \in I} \mpProc{\roleP[i]}{\mpP[i]}$ is shorthand for
$\mpProc{\roleP[1]}{\mpP[1]} \pp \dots \pp \mpProc{\roleP[n]}{\mpP[n]}$,
where $\pp$~is the parallel composition operator, which is associative and commutative.
For simplicity, we frequently assume that the tuples of formal parameters $\tx$ in process definitions and the tuples of
expressions for actual parameters $\widetilde{\mpExpr}$ in process calls are split into standard and quantum ones.

\begin{example}[Quantum teleportation - Processes]
  We can now write an actual system implementing the quantum teleportation protocol:
  \[
    \sysM = \mpProc{\roleAlice}{\mpP[A]} \pp \mpProc{\roleBob}{\mpP[B]} \pp \mpProc{\roleSource}{\mpP[S]} \pp \mpProc{\roleEnv}{\mpP[E]}
  \]
  where:
  \[
    \begin{array}{rcl}
      \mpP[A] & = &
      \mpBranchSingle{}{\roleEnv}{}{w}{
        \mpBranchSingle{}{\roleSource}{}{x}{
          \expUnitary[CNot]{w,x}
          \mpSeq
          \expUnitary[H]{w}
          \mpSeq 
          \expMeasure{\tilde r}{(w,x)}
          \mpSeq
          \mpSel{}{\roleBob}{}{\tilde r}{\mpNil}
      }}                                 \\
      \mpP[B] & = &
      \mpBranchSingle{}{\roleSource}{}{y}{
        \mpBranchSingle{}{\roleAlice}{}{\tilde r}{
          \expUnitary[\sigma_{\tilde r}]{y}
          \mpSeq
      \mpSel{}{\roleEnv}{}{y}{\mpNil}}}  \\
      \mpP[S] & = &
      \expNewQBit[x_s,y_s] \mpSeq
      \expUnitary[H]{x_s} \mpSeq
      \expUnitary[CNot]{x_s,y_s} \mpSeq                                             
      \mpSel{}{\roleAlice}{}{x_s}{
      \mpSel{}{\roleBob}{}{y_s}{\mpNil}} \\
      \mpP[E] & = &
      \mpSel{}{\roleAlice}{}{q}{
        \mpBranchSingle{}{\roleBob}{}{y}{\mpNil[y]}}
    \end{array}
  \]


  $\roleAlice$ takes the desired $\tyQBit$s, then applies to them two unitary transformations: a controlled not $\expUnitaryOp[CNot]$ (which negates the second $\tyQBit$ if the first one is true, just propagates it otherwise -- and acts linearly if the states of the qubits are not classical states but a superposition of them) and a Hadamard gate $\expUnitaryOp[H]$ (which turns classical states into a superposition), measures the two resulting $\tyQBit$s obtaining two classical $\tyBit$s in $r$ (here we extended the $\measureOp$ operator to work on pairs) and sends them to $\roleBob$.

  $\roleBob$ gets the two $\tyBit$s and uses them to apply one of 4 Pauli transformations
  to apply to $y$,
  where $\expUnitaryOp[\sigma_0]$ is the identity matrix,
  and $\expUnitaryOp[\sigma_1], \expUnitaryOp[\sigma_2], \expUnitaryOp[\sigma_3]$
  are the three $2 \times 2$
  Pauli matrices~\footnote{We slightly abuse notation here, we could have modeled this choice in more details using conditionals to actually select one of the transformations.}.
  $\roleSource$ just creates a pair of entangled $\tyQBit$s, and sends its components to, respectively, $\roleAlice$ and $\roleBob$.
  $\roleEnv$ just gives the initial $\tyQBit$ to $\roleAlice$ and gets the result from $\roleBob$.\finex
\end{example}
\noindent
Now we describe the typing of processes and systems.
\begin{definition}[Typing contexts]
  \label{def:typing-contexts}
  $\mpEnv$ (used for process definitions and calls) maps process variables to session types prefixed by the types of the parameters,
  $\stEnv$ maps variables to basic types (except $\tyQBit$s), and
  $\qbEnv$ maps variables and references to (only) the $\tyQBit$ type.


  \smallskip
  \centerline{\(%
    \mpEnv
    \,\coloncolonequals\,
    \mpEnvEmpty
    \bnfsep
    \mpEnv \mpEnvComp\, \mpEnvMap{\mpX}{(\widetilde{\stS};\stT)}
    \quad
    \stEnv
    \,\coloncolonequals\,
    \stEnvEmpty
    \bnfsep
    \stEnv \stEnvComp \stEnvMap{x}{\tyGround}
    \ (\tyGround \neq \tyQBit)
    \quad
    \qbEnv
    \,\coloncolonequals\,
    \qbEnvEmpty
    \bnfsep
    \qbEnv \qbEnvComp \qbEnvMap{z}{\tyQBit}
    \)}%
\end{definition}
Notably, following DILL~\cite{barber1996dual},
we use dual contexts for intuitionistic (classical data) variables and linear (quantum data) variables.
Since we omit delegation, there are no channel variables in contexts $\stEnv$ or $\qbEnv$.
We write $\qbEnv[q]$ when $\qbEnv$ contains only qubit references,
and omit contexts in judgements when empty and unambiguous.

\begin{definition}[Typing Judgements]
  \label{def:typing-judgements}
  Typing judgements are of the form: 

  \smallskip
  \centerline{%
    \(\qstJudge{}{\stEnv}{\qbEnv}{e}{\tyGround}\)
    \quad\quad
    \(\qstJudge{\mpEnv}{\stEnv}{\qbEnv}{\mpP}{\stT}\)
    \quad\quad
    \(\qstJudge{}{\stEnv}{\qbEnv}{\mathcal{M}}{\gtG}\)
  }
\end{definition}
\begin{figure}[t]
  \centerline{\(
    \begin{array}{c}
      \inference[\iruleVar]{
        \stEnvMap{x}{\tyGround} \in \stEnv
      }{
        \qstJudge{}{\stEnv}{\qbEnvEmpty}{x}{\tyGround}
      }
      \qquad
      \inference[\iruleQVar]{
      }{
        \qstJudge{}{\stEnv}{\qbEnvMap{z}{\tyQBit}}{z}{\tyQBit}
      }
      \qquad
      \inference[\iruleBit]{}{
        \stJudge{}{\stEnv}{0,1}{\tyBit}
      }
      \\[1em]
      \inference[\iruleQBit]{
        \qstJudge{\mpEnv}{\stEnv}{\qbEnv \qbEnvComp \qbEnvMap{x}{\tyQBit}}{\mpP}{\stT}
      }{
        \qstJudge{\mpEnv}{\stEnv}{\qbEnv}{\expNewQBit[x] \mpSeq \mpP}{\stT}
      }
      \qquad
      \inference[\iruleMeas]{
        \qstJudge{}{\stEnv}{\qbEnv[1]}{z}{\tyQBit}
        \quad
        \qstJudge{\mpEnv}{\stEnv \stEnvComp \stEnvMap{x}{\tyBit}}{\qbEnv[2]}{\mpP}{\stT}
      }{
        \qstJudge{\mpEnv}{\stEnv}{\qbEnv}{\expMeasure{x}{z} \mpSeq \mpP}{\stT}
      }
      \\[1em]
      \inference[\iruleUnitary]{
        \qstJudge{\mpEnv}{\stEnv}{\qbEnv \qbEnvComp \widetilde{\qbEnvMap{z}{\tyQBit}}}{\mpP}{\stT}
      }{
        \qstJudge{\mpEnv}{\stEnv}{\qbEnv \qbEnvComp \widetilde{\qbEnvMap{z}{\tyQBit}}}{\expUnitary[]{\tz} \mpSeq \mpP}{\stT}
      }
      \qquad
      \inference[\iruleMPNil]{}{
        \qstJudge{\mpEnv}{\stEnv}{\widetilde{\qbEnvMap{z}{\tyQBit}}}{\mpNil[\tz]}{\stEnd}
      }
      \\[1em]
      \inference[\iruleMPSel]{
        \qstJudge{}{\stEnv}{\qbEnv[1]}{e}{\stS}
        \quad
        \qstJudge{\mpEnv}{\stEnv}{\qbEnv[2]}{\mpP}{\stT}
      }{
        \qstJudge{\mpEnv}{\stEnv}{\qbEnv}{\mpSel{}{\roleP}{\mpLab}{e}{\mpP}}{\stIntSum{\roleP}{}{\stChoice{\stLab}{\stS} \stSeq \stT}}
      }
      \\[1em]
      \inference[\iruleMPBranch]{
        \forall i \in I
        \
        \qstJudge{\mpEnv}{\stEnv \stEnvComp \stEnvMap{x_i}{\stS[i]}}{\qbEnv}{\mpP[i]}{\stT[i]}
        \quad
        \forall j \in J
        \
        \qstJudge{\mpEnv}{\stEnv}{\qbEnv \qbEnvComp \qbEnvMap{x_j}{B_j}}{\mpP[j]}{\stT[j]}
      }{
        \qstJudge{\mpEnv}{\stEnv}{\qbEnv}{\mpBranch{}{\roleP}{i \in I \uplus J}{\mpLab[i]}{x_i}{\mpP[i]}{}}{\stExtSum{\roleP}{i \in I \uplus J}{\stChoice{\stLab[i]}{\stS[i]} \stSeq \stT[i]}}
      }
      \\[1em]
      \inference[\iruleIfThen]{
        \qstJudge{}{\stEnv}{\qbEnv[1]}{e}{\tyBit}
        \quad
        \qstJudge{\mpEnv}{\stEnv}{\qbEnv[2]}{\mpP}{\stT}
        \quad
        \qstJudge{\mpEnv}{\stEnv}{\qbEnv[2]}{\mpQ}{\stT}
      }{
        \qstJudge{\mpEnv}{\stEnv}{\qbEnv}{\mpIf{e}{\mpP}{\mpQ}}{\stT}
      }
      \\[1em]
      \inference[\iruleMPIfThen]{
        \qstJudge{}{\stEnv}{\qbEnv[1]}{e}{\tyBit}
        \quad
        \qstJudge{\mpEnv}{\stEnv}{\qbEnv[2]}{\mpP}{\stIntSum{\roleP}{i \in I}{\stChoice{\stLab[i]}{\stS[i]} \stSeq \stT[i]}}
        \quad
        \qstJudge{\mpEnv}{\stEnv}{\qbEnv[2]}{\mpQ}{\stIntSum{\roleP}{i \in J}{\stChoice{\stLab[i]}{\stS[i]} \stSeq \stT[i]}}
      }{
        \qstJudge{\mpEnv}{\stEnv}{\qbEnv}{\mpIf{e}{\mpP}{\mpQ}}{\stIntSum{\roleP}{i \in I \uplus J}{\stChoice{\stLab[i]}{\stS[i]} \stSeq \stT[i]}
        }}
      \\[1em]
      \inference[\iruleMPDef]{
        \qstJudge{\mpEnv \mpEnvComp \mpEnvMap{\mpX}{(\tilde{\stS},\tilde{\stSi};\stT[1])}}{\stEnv \stEnvComp \widetilde{\stEnvMap{x}{\stS}}}{\qbEnv[1] \qbEnvComp \widetilde{\qbEnvMap{x'}{\stSi}}}{\mpP}{\stT[1]}
        \quad
        \qstJudge{\mpEnv \mpEnvComp \mpEnvMap{\mpX}{(\tilde{\stS},\tilde{\stSi};\stT[1])}}{\stEnv}{\qbEnv[2]}{\mpQ}{\stT[2]}
      }{
        \qstJudge{\mpEnv}{\stEnv}{\qbEnv}{\mpDef{\mpX}{\tx,\tx'}{\mpP}{\mpQ}}{\gtRec{\mpX}{\stT[2]}}
      }
      \\[1em]
      \inference[\iruleMPCall]{
        \forall i \in [n]
        \quad
        \qstJudge{}{\stEnv}{\qbEnv[i]}{e_i}{\stS_i}
      }{
        \qstJudge{\mpEnv \mpEnvComp \mpEnvMap{X}{(\widetilde{\stS};\stT)}}{\stEnv}{\qbEnv}{\mpCallSmall{\mpX}{\tilde{e}}}{\stT}
      }
      \qquad
      \inference[\iruleMPSystem]{
        \forall \roleP[i] \in \gtRoles{\gtG}
        \qquad
        \qstJudge{}{\stEnv}{\qbEnv[i]}{\mpP[i]}{\gtProj{\gtG}{\roleP[i]}}
      }{
        \qstJudge{}{\stEnv}{\qbEnv}{\prod\limits_{\roleP[i] \in \gtRoles{\gtG}} \mpProc{\roleP[i]}{\mpP[i]}}{\gtG}
      }
    \end{array}
    \)}%
  \caption{
    Typing rules for expressions, processes, and systems
  }
  \label{fig:mpst-rules}
\end{figure}
\begin{definition}[Typing Rules]
  The typing rules are defined in~\cref{fig:mpst-rules}.
  Typing rules for operators and constants are standard and omitted (but for 0 and 1).
  The $\qbEnv$ used in the conclusion of each typing rule in~\cref{fig:mpst-rules} is the concatenation of the (different) $\qbEnv[i]$s
  used in the premises.
\end{definition}
The rules are mostly self-explanatory. Notably, due to linearity of
$\tyQBit$s, when a process sends a $\tyQBit$ (rule \iruleMPSel) it
loses the ownership (indeed, $\qbEnv$ is split into $\qbEnv[1]$ used
to type the payload of the message and $\qbEnv[2]$ used to type the
continuation). Dually, when a process receives a $\tyQBit$
(rule \iruleMPBranch) it gains ownership on it. Indeed, it can be used
to type the continuation. Note that the two sets of premises are
needed to cover two cases depending on whether the received value is
classical (first case, added to $\stEnv$) or quantum (second case,
added to $\qbEnv$). We have two rules for conditional, \iruleIfThen~applies when the two branches follow the same local type while \iruleMPIfThen~when they contribute with different message send. In both the cases, both the premises use the same $\qbEnv[2]$ since in each computation only one of the branches is executed.
In rule~\iruleMPDef, $\mpEnv$ stores the type of the defined function so that it can then be checked against the invocation in rule~\iruleMPCall. There $\tilde{\stS}$ is shorthand for $\stS[1], \ldots, \stS[n]$, and $\tilde{e}$ is shorthand
for $e_1, \ldots, e_n$.

\begin{example}[Quantum teleportation - Typing]
  Using the typing rules we can derive the judgment:
  \[
    \qstJudge{}{\stEnvEmpty}{\stEnvMap{q}{\tyQBit}}{\sysM}{\gtG}
  \]
  We show the typing of $\roleAlice$ as an example in~\cref{fig:qtAlice}.\finex
  \begin{figure*}[t]
    \scalebox{\scalef}{
      \inference[]{
        \inference[]{
          \inference[]{
            \inference[]{
              \inference[]{
                \inference[]{
                  \stJudge{}{\stEnvMap{\tilde r}{\tyBit^2}}{
                    \tilde r}
                  { \tyBit^2}
                  \qquad
                  \stJudge{\emptyset}{\stEnvMap{\tilde r}{\tyBit^2}}{
                    \mpNil}
                  { \stEnd}
                }{
                  \stJudge{\emptyset}{\stEnvMap{\tilde r}{\tyBit^2}}{
                    \mpSel{}{\roleBob}{}{\tilde r}{\mpNil}}
                  { \stIntSumSing{\roleFmt{Bob}}{}{\stChoice{}{\tyBit^2} \stSeq \stEnd}}
                }}{
                \stJudge{\emptyset}{\stEnvMap{w}{\tyQBit},\stEnvMap{x}{\tyQBit}}{
                  \expMeasure{\tilde r}{(w,x)}
                  \mpSeq
                  \mpSel{}{\roleBob}{}{\tilde r}{\mpNil}}
                { \stIntSumSing{\roleFmt{Bob}}{}{\stChoice{}{\tyBit^2} \stSeq \stEnd}}
              }}{
              \stJudge{\emptyset}{\stEnvMap{w}{\tyQBit},\stEnvMap{x}{\tyQBit}}{
                \expUnitary[H]{w}
                \mpSeq
                \expMeasure{\tilde r}{(w,x)}
                \mpSeq
                \mpSel{}{\roleBob}{}{\tilde r}{\mpNil}}
              { \stIntSumSing{\roleFmt{Bob}}{}{\stChoice{}{\tyBit^2} \stSeq \stEnd}}
            }
          }{
            \stJudge{\emptyset}{\stEnvMap{w}{\tyQBit},\stEnvMap{x}{\tyQBit}}{
              \expUnitary[CNot]{w,x}
              \mpSeq
              \expUnitary[H]{w}
              \mpSeq
              \expMeasure{\tilde r}{(w,x)}
              \mpSeq
              \mpSel{}{\roleBob}{}{\tilde r}{\mpNil}}
            { \stIntSumSing{\roleFmt{Bob}}{}{\stChoice{}{\tyBit^2} \stSeq \stEnd}}
          }}
        {
          \stJudge{\emptyset}{\stEnvMap{w}{\tyQBit}}{
            \mpBranchSingle{}{\roleSource}{}{x}{
              \expUnitary[CNot]{w,x}
              \mpSeq
              \expUnitary[H]{w}
              \mpSeq
              \expMeasure{\tilde r}{(w,x)}
              \mpSeq
              \mpSel{}{\roleBob}{}{\tilde r}{\mpNil}}}
          { \stExtSumSing{\roleFmt{Source}}{}{\stChoice{}{\tyQBit} \stSeq \stIntSumSing{\roleFmt{Bob}}{}{\stChoice{}{\tyBit^2} \stSeq \stEnd}}}}
      }{
        \stJudge{\emptyset}{\emptyset}{\mpBranchSingle{}{\roleEnv}{}{w}{
            \mpBranchSingle{}{\roleSource}{}{x}{
              \expUnitary[CNot]{w,x}
              \mpSeq
              \expUnitary[H]{w}
              \mpSeq
              \expMeasure{\tilde r}{(w,x)}
              \mpSeq
              \mpSel{}{\roleBob}{}{\tilde r}{\mpNil}}}}{\gtProj{\gtG}{\roleFmt{Alice}}}
      }
    }
    \caption{Typing for Alice in Quantum Teleportation}\label{fig:qtAlice}
  \end{figure*}
\end{example}

\section{Semantics}
\label{sec:semantics}

The semantics that we present combines the management of quantum state
following the approach in~\cite{dal2015applicative,SelingerV06} (for a
quantum $\lambda$-calculus), with classical semantics for MPST systems,
such as in~\cite{gentleIntro,ScalasY19}. A main
novelty is that we need to combine non-determinism (typical of process
calculi but absent in $\lambda$-calculi) with probabilities
originating from measurements of quantum states.

%
Given a finite set of qubit variables $\qregone$, called a quantum register,
we denote with $\hilb{\qregone}$ the Hilbert space generated by all functions $\qregone \to \{0,1\}$, denoted as $\basis{\qregone}$,
which is the computational basis for $\hilb{\qregone}$.
Elements of $\basis{\qregone}$, are of the form $\ket{q_1 \leftarrow b_1, \ldots, q_n \leftarrow b_n}$
and ranged over by~$\ket{\eta}$,
for every ${q_i \in \qregone}$ and ${b_i \in \{0,1\}}$.
Hence, elements of $\hilb{\qregone}$ are elements of Hilbert spaces with a chosen basis,
that is, linear combinations of elements of $\basis{\qregone}$,
written $\qbregone = \sum_{\ket{\eta} \in \basis{\qregone}}{\alpha_{\eta}}{\ket{\eta}}$,
where ${\alpha_{\eta} \in \cmplx}$ are complex probability amplitudes
satisfying $\sum_{\ket{\eta} \in \basis{\qregone}} \abs{\alpha_{\eta}}^2 = 1$.
The probability operator~$\qprob_b^q \colon {\hilb{\qregone} \to \rint}$ is given by
\(
\qprob_b^q(\qbregone) = \sum_{\eta(q) = b} \abs{\alpha_{\eta}}^2
\).
For ${q \in \qregone}$, the projection operator~$\qproj_b^q \colon \hilb{\qregone} \to \hilb{\setsub{\qregone}{q}}$
is given by
\(
\qproj_b^q(\qbregone) =
\sfrac{1}{\sqrt{\qprob_b^q(\qbregone)}}
\sum_{\ket{\eta} \in \basis{\setsub{\qregone}{q}}} {\alpha_{\eta\{q\leftarrow b\}}}{\ket{\eta}}
\).
For ${q \not\in \qregone}$, the new operator~$\qnew_q \colon \hilb{\qregone} \to \hilb{\setadd{\qregone}{q}}$ is given by
\(
\qnew_q(\qbregone) = \qbregone \otimes \ket{q \leftarrow 0}
\)
which assigns $\ket{0}$ to $q$.
Qubits can be modified by unitary operations, written as $U_{\tq}\colon \hilb{\qregone} \to \hilb{\qregone}$,
where $\tq$ is a sequence of distinct qubit variables.


The semantics, defined in~\cref{fig:concredsem}, is a transition system over configurations of the form $\pair{\qbregone}{\sysone}$,
where $\qbregone$ is the quantum state associated to the underlying quantum register $\qregone$,
and $\sysone$ is a multiparty session-typed system.
The semantics is both non-deterministic and probabilistic
(it forms a GPLTS, in the sense of~\cite[Def~2.5]{bernardo2013uniform}) --
the multiparty process behavior is non-deterministic
(in the choice of redex for single processes which do some computation, or pairs of processes that communicate),
and the quantum operations are probabilistic
(in the choice of measurement outcome). 
The transition relation is:
\(
\longrightarrow~{\subseteq}~\concconf \times \labels \times \rint \times \concconf
\),
where $\labels$ is the set of labels: either $\trans{\roleP}{\roleQ}{v}$ denoting a communication step,
or $\roleP$ denoting a computation step (for role names $\roleP$ and $\roleQ$).
We use $\prob$ to range over probabilities.
The semantics relies on a semantics of expressions as well. We do not present it since it is classical,
and write $\expEval{e}{v}$ to denote that expression $e$ evaluates to value $v$.
The evaluation relation on tuples $\expEval{\te}{\td}$ is taken pointwise.

Rule~\inferrule{\ruleRedComm} allows processes to communicate.
Rules~\inferrule{\ruleRedMeasZero} and~\inferrule{\ruleRedMeasOne} define measurements,
and are the only rules defining transitions with non-1 probability.
They also consume the measured qubit.
Dually, rule~\inferrule{\ruleRedNew} extends the register with a new qubit.
Note that there are no transitions for $\mpNil[\tq]$ processes.
%
%
%
\begin{figure*}[t]
	\centerline{\(
		\begin{array}{c}
			\inference[\ruleRedComm]
			{ j \in I \qquad \mpLab = \mpLab[j] \qquad \expEval{e}{v}}
			{
				\pair{\qbregone}{\para{\named{\roleP}{\mpSel{}{\roleQ}{\mpLab}{e}{\mpP}}}{\named{\roleQ}{\mpBranch{}{\roleP}{i \in I}{\mpLab[i]}{x_i}{\mpQ[i]}{}}}}
				\plts{\trans{\roleP}{\roleQ}{v}}{}
				\pair{\qbregone}{\para{\named{\roleP}{P}}{\named{\roleQ}{\mpQ[j]\subst{x_j}{v}}}}
			}
			\\[1em]
			\inference[\ruleRedDef]
			{\expEval{\te}{\td}}
			{
				\pair{\qbregone}{\named{\roleP}{\mpDef{\mpX}{\tx}{\mpP}{\mpQ}}}
				\plts{\roleP}{}
				\pair{\qbregone}{\named{\roleP}{Q\subst{\mpX(\te)}{P\subst{\tx}{\td}}}}
			}
			\\[1em]
			\inference[\ruleRedNew]{
				x \not\in \qregone
			}{
				\pair{\qbregone}{\mpProc{\roleP}{\expNewQBit[x] \mpSeq \mpP}}
				\plts{\roleP}{}
				\pair{\qnew_x(\qbregone)}{\mpProc{\roleP}{\mpP}}
			}
			\quad
			\inference[\ruleRedUnitary]
			{ \tz\subseteq\qregone }
			{
				\pair{\qbregone}{\named{\roleP}{\expUnitary[]{\tz}.P}}
				\plts{\roleP}{}
				\pair{\expUnitaryOp[]_{\tz}(\qbregone)}{\named{\roleP}{P}}
			}
			\\[1em]
			\inference[\ruleRedMeasZero]
			{ z \in\qregone \qquad \prob = \qprob_0^z(\qbregone) }
			{
				\pair{\qbregone}{\named{\roleP}{\measure{x}{z}.P}}
				\plts{\roleP}{\prob}
				\pair{\qproj_0^z(\qbregone)}{\named{\roleP}{P\subst{x}{0}}}
			}
			\\[1em]
			\inference[\ruleRedMeasOne]
			{ z \in\qregone \qquad \prob = \qprob_1^z(\qbregone) }
			{
				\pair{\qbregone}{\named{\roleP}{\measure{x}{z}.P}}
				\plts{\roleP}{\prob}
				\pair{\qproj_1^z(\qbregone)}{\named{\roleP}{P\subst{x}{1}}}
			}
			\\[1em]
			\inference[\ruleRedIfTrue]
			{ \expEval{e}{1} }
			{
				\pair{\qbregone}{\named{\roleP}{\mpIf{e}{\mpP}{\mpQ}}}
				\plts{\roleP}{}
				\pair{\qbregone}{\named{\roleP}{\mpP}}
			}
			\quad
			\inference[\ruleRedIfFalse]
			{ \expEval{e}{0} }
			{
				\pair{\qbregone}{\named{\roleP}{\mpIf{e}{\mpP}{\mpQ}}}
				\plts{\roleP}{}
				\pair{\qbregone}{\named{\roleP}{\mpQ}}
			}
			\\[1em]
			\inference[\ruleRedCong]
			{
				\pair{\qbregone}{\sysone}
				\plts{\lblComm}{\prob}
				\pair{\qbregtwo}{\systwo}
			}
			{
				\pair{\qbregone}{\para{\sysone}{\mpP}}
				\plts{\lblComm}{\prob}
				\pair{\qbregtwo}{\para{\systwo}{\mpP}}
			}
		\end{array}
		\)}
	\caption{Semantic Rules}\label{fig:concredsem}
\end{figure*}
The transition system satisfies the following properties (expected for GPLTS):
\begin{propositionrep}
	\label{prop:concrete-semantics}
	For any $\qbregone, \sysone, \qbregtwo, \systwo, \alpha$,
	\begin{enumerate}[leftmargin=*]
		\item if $\pair{\qbregone}{\sysone} \plts{\alpha}{\prob_1} \pair{\qbregtwo}{\systwo}$
		      and $\pair{\qbregone}{\sysone} \plts{\alpha}{\prob_2} \pair{\qbregtwo}{\systwo}$,
		      then $\prob_1 = \prob_2$, and
		\item it holds that
		      \(
		      {\sum\set{\prob \in \rint | \pair{\qbregone}{\sysone} \plts{\alpha}{\prob} \pair{\qbregtwo}{\systwo}} = 1}
		      \).
	\end{enumerate}
\end{propositionrep}
\begin{proof}
	One can notice that the only probabilities which are not $1$ originate
	from measurements. However, the first part holds since the two
	transitions corresponding to the two outcomes of a given measurement
	always lead to different states. This is also the only case where there are two transitions from the same state with the same label, and since the probabilities sum up to $1$ also the second part holds.
\end{proof}
We remark that labels in our semantics are just decorations to
simplify the statement of the results in the next section, hence the
semantics can be seen as an annotated reduction semantics.

Note that we could also give a more abstract description of the semantics, disregarding the
qubit register and probabilities.
In this setting, actions from a single participant are to be considered as invisible $\lblComp$ steps.
This allows one to more easily relate with MPSTs which indeed only describe communications,
but we choose to work the more concrete \emph{quantum} semantics, showing the quantum state and probabilities.
\begin{definition}[Multiparty System Transitions]
	We say $\sysone \lts{\lblComm} \systwo$ if $\exists \qbregone, \qbregtwo, \prob \neq 0$
	such that $\pair{\qbregone}{\sysone} \plts{\lblComm}{\prob} \pair{\qbregone}{\systwo}$.
	We say $\sysone \lts{\lblComp} \systwo$ if $\exists \roleP$,
	such that $\sysone \lts{\roleP} \systwo$.
	We say $\sysM \Lts{\lblComm} \sysMi$ if $\sysM \ltsStar{\ps{\lts{\lblComp}}} \,\lts{\lblComm}{}\, \ltsStar{\ps{\lts{\lblComp}}} \sysMi$.
\end{definition}


\section{Results}
\label{sec:results}




We prove two kinds of results -- the first kind are classical results adapted
from MPSTs theory, such as Subject Reduction and Session Fidelity (see,
e.g.,~\cite{gentleIntro,ScalasY19,CASTELLANI2023100844}), to the quantum extension.
Such results hold for \emph{closed} terms, namely terms with no free variables,
but, notably, they \emph{can have references} to the qubit register.
Hence the judgments in these statements have the form $\qstJudge{}{\stEnvEmpty}{\qbEnv[q]}{\sysM}{\gtG}$.
The other kind of results show that
$\tyQBit$s are treated correctly, linearly in particular.
\ifmain
Details of proofs and auxiliary results can be found in~\cref{apx:results}.
\fi

\begin{toappendix}
   \label{apx:results}
   \begin{lemma}
      \label{lem:contractive}
      If a global type $\gtG$ is contractive, then the projection $\gtProj{\gtG}{\roleP}$,
      for any $\roleP \in \gtRoles{\gtG}$, is contractive.
   \end{lemma}
   \begin{proof}
      For a contractive type, the recursive variable always appears under a type constructor.
      The result follows from inspecting the definition of global type projection in the recursive case.
   \end{proof}


   \begin{lemma}[Substitution lemmas]
      \label{lem:substitution}
      The following substitution lemmas hold:
      \begin{enumerate}[leftmargin=*]
         \item If $\qstJudge{}{\stEnv}{\qbEnv\qbEnvComp\qbEnvMap{z}{\tyQBit}}{e}{\tyGround}$
               and $\qstJudge{}{\stEnv}{\qbEnvi}{z'}{\tyQBit}$,
               then $\qstJudge{}{\stEnv}{\qbEnv\qbEnvComp\qbEnvi}{e\subst{z}{z'}}{\tyGround}$.

         \item If $\qstJudge{}{\stEnv\stEnvComp\stEnvMap{x}{\tyGround[1]}}{\qbEnv}{e}{\tyGround[2]}$
               and $\qstJudge{}{\stEnv}{\qbEnvEmpty}{v}{\tyGround[1]}$,
               then $\qstJudge{}{\stEnv}{\qbEnv}{e\subst{x}{v}}{\tyGround[2]}$.

         \item If $\qstJudge{\mpEnv\mpEnvComp\mpEnvMap{\mpX}{(\tilde{\stS},\tilde{\stSi};\stT[1])}}{\stEnv\stEnvComp\widetilde{\stEnvMap{x}{\stS}}}{\qbEnv\qbEnvComp \widetilde{\qbEnvMap{z}{\stSi}}}{\mpP}{\stT[1]}$,
               $\qstJudge{\mpEnv\mpEnvComp\mpEnvMap{\mpX}{(\tilde{\stS},\tilde{\stSi};\stT[1])}}{\stEnv}{\qbEnvi}{\mpQ}{\stT[2]}$,
               and $\qstJudge{}{\stEnv}{\qbEnvEmpty}{d_i}{\stS[i]}$ and $\qstJudge{}{\stEnv}{\qbEnvii[i]}{d'_i}{\stSi[i]}$,
               then $\qstJudge{\mpEnv}{\stEnv}{\qbEnv\qbEnvComp\qbEnvi\qbEnvComp\qbEnvii}{\mpQ\subst{\mpX(\tilde{\mpD},\tilde{\mpDi})}{P\subst{\tx}{\tilde{\mpD}}\subst{\tq}{\tilde{\mpDi}}}}{\stT}$.

         \item If $\qstJudge{\mpEnv}{\stEnv}{\qbEnv\qbEnvComp\qbEnvMap{z}{\tyQBit}}{\mpP}{\stT}$
               and $\qstJudge{}{\stEnv}{\qbEnvi}{z'}{\tyQBit}$,
               then $\qstJudge{\mpEnv}{\stEnv}{\qbEnv\qbEnvComp\qbEnvi}{\mpP\subst{z}{z'}}{\stT}$.

         \item If $\qstJudge{\mpEnv}{\stEnv \stEnvComp \stEnvMap{x}{\tyGround}}{\qbEnv}{\mpP}{\stT}$
               and $\qstJudge{}{\stEnv}{\qbEnvEmpty}{v}{\tyGround}$,
               then $\qstJudge{\mpEnv}{\stEnv}{\qbEnv}{\mpP\subst{x}{v}}{\stT}$.
      \end{enumerate}
   \end{lemma}
   \begin{proof}
      By induction on the derivation of the first judgement,
      where the interesting cases are the binding forms in the rules~\inferrule{\iruleQBit}~and~\inferrule{\iruleMeas}.
   \end{proof}

\end{toappendix}




First, any global type which is not $\gtEnd$ can perform at least a transition.
\begin{lemmarep}[Progress for global types]\label{lem:progress-global-types}
   If $\gtG \neq \gtEnd$, then $\gtG \,\gtMove$.
\end{lemmarep}
\begin{proof}
   By inspection of the rule~\inferrule{\iruleGtMoveComm},
   and noting that label sets are non-empty.
\end{proof}


\begin{toappendix}

   \begin{lemma}[Projection lemmas]
      \label{lem:projection}
      \leavevmode
      \begin{enumerate}[leftmargin=*]
         \item\label[part]{lem:projection:1}
               If $\gtG \,\gtMove[\lblComm]\, \gtGi$, then $\gtRoles{\gtG} \supseteq \gtRoles{\gtGi}$.
         \item\label[part]{lem:projection:2}
               If $\roleP, \roleQ \in \gtRoles{\gtG}$,
               such that
               ${\gtProj{\gtG}{\roleP} = \stIntSum{\roleQ}{i \in I}{\stChoice{\stLab[i]}{\tyGround[i]} \stSeq \stT[i]}}$ and
               ${\gtProj{\gtG}{\roleQ} = \stExtSum{\roleP}{i \in J}{\stChoice{\stLabi[i]}{\tyGroundi[i]} \stSeq \stTi[i]}}$,
               then $I = J$, and for every $i \in I$, we have
               $\stLab[i] = \stLabi[i]$, $\tyGround[i] = \tyGroundi[i]$,
               $\gtG \gtMove[\gtCommLab{\roleP}{\roleQ}{}:{{\gtLab[i]}({\tyGround[i]})}]\, \gtG[i]$,
               $\gtProj{\gtG[i]}{\roleP} = \stT[i]$, and $\gtProj{\gtG[i]}{\roleQ} = \stTi[i]$.
         \item\label[part]{lem:projection:3}
               If $\gtG \gtMove[\gtCommLab{\roleP}{\roleQ}{}:{{\gtLab[j]}({\tyGround[j]})}]\, \gtG[j]$
               for some $j \in I$, then we have
               ${\gtProj{\gtG}{\roleP} = \stIntSum{\roleQ}{i \in I}{\stChoice{\stLab[i]}{\tyGround[i]} \stSeq \stT[i]}}$,
               ${\gtProj{\gtG}{\roleQ} = \stExtSum{\roleP}{i \in I}{\stChoice{\stLab[i]}{\tyGround[i]} \stSeq \stTi[i]}}$,
               and ${\gtProj{\gtG[j]}{\roleR} = \gtProj{\gtG}{\roleR}}$, for ${\roleR \not\in \set{\roleP, \roleQ}}$.
      \end{enumerate}
   \end{lemma}
   \begin{proof}
      By analysing the definition of global type projection.
   \end{proof}

\end{toappendix}

\noindent
Subject reduction, stated below, ensures that transitions preserve typability for systems with no free variables (which may however have references to the qubit register).
We write $\exists\,\qbEnv\wk\qbEnvi$ to mean that
there exists a context $\qbEnvi$ such that $\qbEnv$ is a weakening of $\qbEnvi$.

\begin{theoremrep}[Subject Reduction]
   \label{thm:subject-reduction}
   \leavevmode
   \begin{enumerate}[leftmargin=*]
      \item If $\qstJudge{}{\stEnvEmpty}{\qbEnv[q]}{\sysM}{\gtG}$ and $\sysM \lts{\roleP} \sysMi$,
            then $\exists\qbEnvi[q]$ such that
            $\qstJudge{}{\stEnvEmpty}{\qbEnvi[q]}{\sysMi}{\gtG}$.

      \item If $\qstJudge{}{\stEnvEmpty}{\qbEnv[q]}{\sysM}{\gtG}$ and $\sysM \lts{\trans{\roleP}{\roleQ}{\mpLab(v)}} \sysMi$,
            then $\exists\,\qbEnv[q]\wk\qbEnvi[q],\tyGround,\gtGi$
            such that
            $\qstJudge{}{\stEnvEmpty}{\qbEnvi[q]}{v}{\tyGround}$,
            ${\gtG \,\gtMove[\trans{\roleP}{\roleQ}{\mpLab(\tyGround)}]\, \gtGi}$,
            and $\qstJudge{}{\stEnvEmpty}{\qbEnv[q]}{\sysMi}{\gtGi}$.
   \end{enumerate}
\end{theoremrep}
\begin{proof}
   By case analysis on the transition label.
   Part 1:
   By inverting the $\roleP$ transition derivation, then inversion on the typing derivation,
   and substitution~\lcnamecrefs{lem:substitution}~(\cref{lem:substitution}).
   Note that this lemma holds for closed terms (no free variables except for qubit references)
   because of how the $\lts{}{}$ relation is defined, tracking the qubit environment in the quantum register.
   Part 2:
   This is mostly standard following~\cite[theorem 6.10]{CASTELLANI2023100844},
   and using the projection lemmas~(\cref{lem:projection}).
   Note that if $\mpV$ is a reference to a qubit register, then it needs to be typed in $\qbEnvi[q]$,
   and that $\roleP$ sends it, while $\roleQ$ receives it, hence it still belongs to $\sysMi$,
   making $\qbEnv[q]$ stay the same.
\end{proof}

Note that the environment $\qbEnv[q]$ never changes during communication
steps: ownership can move from one process to the other, but
$\tyQBit$s are neither created nor destroyed.
They can only be created and destroyed in computational steps, using, respectively, new qubit generation and measurement.

Session fidelity shows that transitions of global types and of systems
do agree. Notice however that global types abstract away from non-communication
steps, which are instead present in systems.
\begin{theoremrep}[Session Fidelity]
   \label{thm:session-fidelity}
   \leavevmode
   \begin{enumerate}[leftmargin=*]
      \item If $\qstJudge{}{\stEnvEmpty}{\qbEnv[q]}{\sysM}{\gtG}$ and
            ${\gtG \,\gtMove[\gtCommLab{\roleP}{\roleQ}{}:{{\gtLab}({\tyGround})}]\, \gtGi}$,
            then $\exists\,\qbEnvi[q],\qbEnvii[q],\sysMi,\mpV$,
            such that $\qstJudge{}{\stEnvEmpty}{\qbEnvi[q]}{\mpV}{\tyGround}$,
            ${\sysM \Lts{\trans{\roleP}{\roleQ}{\mpLab(\mpV)}} \sysMi}$,
            and $\qstJudge{}{\stEnvEmpty}{\qbEnvii[q]}{\sysMi}{\gtGi}$.

      \item If $\qstJudge{}{\stEnvEmpty}{\qbEnv[q]}{\sysM}{\gtG}$,
            ${\sysM \Lts{\trans{\roleP}{\roleQ}{\mpLab(\mpV)}} \sysMi}$,
            and ${\gtG \,\gtMove[\trans{\roleP}{\roleQ}{\mpLab(\tyGround)}]\, \gtGi}$,
            then $\exists\,\qbEnvi[q],\qbEnvii[q]$,
            such that $\qstJudge{}{\stEnvEmpty}{\qbEnvi[q]}{\mpV}{\tyGround}$,
            and $\qstJudge{}{\stEnvEmpty}{\qbEnvii[q]}{\sysMi}{\gtGi}$.
   \end{enumerate}
\end{theoremrep}
\begin{proof}
   Part 1:
   We do an inversion on the~\inferrule{\iruleGtMoveComm} rule, hence there exists $j \in I$ such that
   \( \beta = \gtCommLab{\roleP}{\roleQ}{}:{{\gtLab[j]}({\tyGround[j]})} \),
   and the type $\gtG$ becomes
   \( \gtComm{\roleP}{\roleQ}{i \in I}{\gtLab[i]}{\tyGround[i]}{\gtG[i]} \),
   which types $\sysM$ in $\qbEnv[q]$.
   By~\cref{lem:projection}~part 3, we get the projections of $\gtG$.
   By induction on the typing derivation of $\roleP$ and $\roleQ$ we get that they take complementary communication actions, possibly preceded by internal steps.
   We can lift all the steps to $\sysM$ using rule~\ruleRedCong, preceded by rule~\ruleRedComm~for the communication step, ontaining $\sysMi$.
   By~\cref{thm:subject-reduction} there exists $\qbEnvii[q]$ such that $\qstJudge{}{\stEnvEmpty}{\qbEnvii[q]}{\sysMi}{\gtGi}$.
   Part 2: Similarly, and by using~\cref{thm:subject-reduction} on the reduction of $\sysM$.
\end{proof}

\noindent
We can now show that well-typed systems have progress and never get stuck.
\begin{theoremrep}[Progress]\label{thm:progress}
   If $\qstJudge{}{\stEnvEmpty}{\qbEnv[q]}{\sysM}{\gtG}$, then either $\gtG = \gtEnd$, or $\sysM \lts{}{}$.
\end{theoremrep}
\begin{proof}
   By~\cref{thm:session-fidelity} and~\cref{lem:progress-global-types}.
\end{proof}

\begin{definition}\label{def:stuck}
   A system $\sysM$ is stuck if it is not of the form $\prod_{i \in I} \mpProc{\roleP[i]}{\mpNil[\tz_i]}$,
   and $\sysM \not\lts{}{}$.
   A system $\sysM$ gets stuck, written $\stuck{\sysM}$, if it reduces to a stuck system.
\end{definition}

\begin{theoremrep}[Type Safety]\label{thm:safety}
   If $\qstJudge{}{\stEnvEmpty}{\qbEnv[q]}{\sysM}{\gtG}$, then $\stuck{\sysM}$ does not hold.
\end{theoremrep}
\begin{proof}
   By~\cref{thm:subject-reduction,thm:progress}.
\end{proof}

Given a well-typed process $\qstJudge{}{\stEnv}{\qbEnv}{\mpP}{\stT}$,
the free qubits of $\mpP$, written $\fq{\mpP}$ are exactly the qubit variables in $\qbEnv$.
The typing rule~\inferrule{\iruleMPSystem} ensures unique ownership of qubits in a multiparty system.
Alternatively, given a well-typed multiparty system
$\qstJudge{}{\stEnv}{\qbEnv}{\prod_{\roleP[i] \in \gtRoles{\gtG}} \mpProc{\roleP[i]}{\mpP[i]}}{\gtG}$,
$\fq{\mpP[i]}$ can be computed as $\qbEnv \cap \fv{\mpP[i]}$
(see~\cref{def:free-variables}
\ifmain
in~\cref{apx:def:free-variables}
\fi
),
which is exactly $\qbEnv[i]$, by inverting the typing rule~\inferrule{\iruleMPSystem}.
\begin{toappendix}
   \label{apx:def:free-variables}
   \begin{definition}[Free variables]
      \label{def:free-variables}
      Function $\fv{\cdot}$ computes the free variables (and references) on systems $\sysone$ and processes $\mpP$. It is defined recursively as follows:
      \[
         \begin{array}{rcl}
            \fv{\prod_{i \in I} \mpProc{\roleP[i]}{\mpP[i]}}             & = & \bigcup_i \fv{\mpP[i]}                        \\[0.5em]
            \fv{\mpSel{}{\roleP}{\mpLab}{e}{\mpP}}                       & = & \fv{e} \cup \fv{\mpP}                         \\[0.5em]
            \fv{\mpBranch{}{\roleP}{i \in I}{\mpLab[i]}{x_i}{\mpP[i]}{}} & = & \bigcup_{i \in I}(\setsub{\fv{\mpP[i]}}{x_i}) \\[0.5em]
            \fv{\mpDef{\mpX}{\tx}{\mpP}{\mpQ}}                           & = & \setsub{\fv{\mpP}}{\tx} \cup \fv{\mpQ}        \\[0.5em]
            \fv{\mpCallSmall{\mpX}{\te}}                                 & = & \bigcup_{e \in \te} \fv{e}                    \\[0.5em]
            \fv{\expMeasure{x}{z}.\mpP}                                  & = & \setsub{(\sets{z} \cup \fv{\mpP})}{x}                       \\[0.5em]
            \fv{\expNewQBit[x].\mpP}                                     & = & \setsub{\fv{\mpP}}{x}                         \\[0.5em]
            \fv{\expUnitary[]{\tz}.\mpP}                                 & = & \tz \cup \fv{P}                               \\[0.5em]
            \fv{\mpIf{e}{\mpP}{\mpQ}}                                    & = & \fv{e} \cup \fv{\mpP} \cup \fv{\mpQ}          \\[0.5em]
            \fv{\mpNil[\tz]}                                             & = & \tz                                           \\[0.5em]
         \end{array}
      \]
   \end{definition}
\end{toappendix}
We can rephrase in our setting the result in \cite[Theorem 2]{GayN05}.

\begin{theoremrep}[Unique Ownership of Qubits]\label{thm:unique}
   If $\qstJudge{}{\stEnv}{\qbEnv}{\prod_{i \in I} \mpProc{\roleP[i]}{\mpP[i]}}{\gtG}$,
   then for any $j, k \in I$ such that $j \neq k$, we have $\fq{P_j} \cap \fq{P_k} = \emptyset$.
\end{theoremrep}
\begin{proof}
   This follows by observing the $\inferrule{\iruleMPSystem}$ rule,
   which uses the concatenation of each $\qbEnv[i]$ in $\qbEnv$ in the conclusion.
\end{proof}
\noindent
By safety, we can extend the transitions to typed systems as follows.
\begin{definition}[Typed Multiparty System Transitions]
   We say that $\qstJudge{}{\stEnvEmpty}{\qbEnv[q]}{\sysM}{\gtG} \Lts{\alpha} \qstJudge{}{\stEnvEmpty}{\qbEnvi[q]}{\sysMi}{\gtGi}$ if
   $\sysM \Lts{\lblComm} \sysMi$ and $\gtG \lts{\lblCommi} \gtGi$ and $\qstJudge{}{\stEnvEmpty}{\qbEnvi[q]}{\sysMi}{\gtGi}$.
\end{definition}
\noindent
We can thus show that the $\tyQBit$s needed at runtime are only the ones prescribed by the typing judgment.

\begin{toappendix}

   \begin{lemma}[Safety of Qubit Registers]
      \label{lem:qubit-monotonicity}
      If $\qstJudge{}{\stEnvEmpty}{\qbEnv[q]}{\sysM}{\gtG}$ and $\qstJudge{}{\stEnvEmpty}{\qbEnvi[q]}{\sysMi}{\gtGi}$,
      and $\pair{\qbregone}{\sysM} \lts{\lblComm}{} \pair{\qbregtwo}{\sysMi}$,
      then $\ps{\qregone {\setminus} \fq{\sysM}} = \ps{\qregtwo {\setminus} \fq{\sysMi}}$.
   \end{lemma}
   \begin{proof}
      By case analysis on the reduction relation,
      and using the assumptions on the quantum registers for each case, in particular the $\lblComp$ transitions.
   \end{proof}

\end{toappendix}

\begin{theoremrep}[Qubit Safety]
   \label{thm:qubit-safety}
   If $\qstJudge{}{\stEnvEmpty}{\qbEnv[q]}{\sysM}{\gtG} \Lts{\alpha} \qstJudge{}{\stEnvEmpty}{\qbEnvi[q]}{\sysMi}{\gtGi}$,
   then $\pair{\hilb{\qu{\qbEnv[q]}}}{\sysM} \Lts{\alpha} \pair{\hilb{\qu{\qbEnvi[q]}}}{\sysMi}$,
   where $\qu{\qbEnv} = \set{q | \qbEnvMap{q}{\tyQBit} \in \qbEnv}$.
\end{theoremrep}
\begin{proof}
   By~\cref{lem:qubit-monotonicity,thm:subject-reduction,thm:session-fidelity}.
\end{proof}


\section{Case Studies}
\label{sec:examples}

We now put the theory developed in the previous sections to work on
some sample quantum protocols from the literature, to complement
Quantum Teleportation described in the previous sections. All the
protocols mix communications of classical values and communications of
references to qubits.

\begin{example}[Quantum secret sharing {\cite[Fig.~21]{gay18}}]\label{ex:secret}
    The problem of secret sharing~\cite{secretSharing} involves an agent $\roleAlice$ sending
    a $\tyQBit$ to two agents $\roleBob$ and $\roleCharlie$, one of whom is dishonest. $\roleAlice$ does not know which is the dishonest agent, so she encodes the $\tyQBit$ in a way that requires $\roleBob$ and $\roleCharlie$ to collaborate to retrieve it. In order to do this, they need to have each a $\tyQBit$ which is part of a maximally entangled triple, provided by $\roleSource$.
    As in the running example, we introduce $\roleEnv$ to model input/output operations.
    \begin{multline*}\trans{\roleEnv}{\roleAlice}{\tyQBit} \gtSeq \trans{\roleSource}{\roleAlice}{\tyQBit} \gtSeq \trans{\roleSource}{\roleBob}{\tyQBit} \gtSeq\\ \trans{\roleSource}{\roleCharlie}{\tyQBit} \gtSeq \trans{\roleAlice}{\roleCharlie}{\tyBit^2} \gtSeq\\ \trans{\roleBob}{\roleCharlie}{\tyBit} \gtSeq \trans{\roleCharlie}{\roleEnv}{\tyQBit}
    \end{multline*}
    We can now write an actual system implementing this multiparty global type:
    \[
        \sysM = \mpProc{\roleAlice}{\mpP[A]} \pp \mpProc{\roleBob}{\mpP[B]} \pp \mpProc{\roleCharlie}{\mpP[C]} \pp \mpProc{\roleSource}{\mpP[S]} \pp \mpProc{\roleEnv}{\mpP[E]}
    \]
    where:
    \[
        \begin{array}{rcl}
            \mpP[A] & = & \mpBranchSingle{}{\roleEnv}{}{x}{
                \mpBranchSingle{}{\roleSource}{}{a}{
                    \!
                    \begin{array}[t]{l}
                        \expUnitary[CNot]{x,a}
                        \mpSeq
                        \expUnitary[H]{x}
                        \mpSeq
                        \expMeasure{m}{x}
                        \mpSeq \\
                        \expMeasure{n}{a}
                        \mpSeq
                        \mpSel{}{\roleCharlie}{}{m,n}{\mpNil}
                    \end{array}
                }
            }                                                                            \\
            \mpP[B] & = & \mpBranchSingle{}{\roleSource}{}{b}{
                \expUnitary[H]{b}
                \mpSeq
                \expMeasure{o}{b}
                \mpSeq
                \mpSel{}{\roleCharlie}{}{o}{\mpNil}
            }                                                                            \\
            \mpP[C] & = & \begin{array}[t]{l}
                              \mpBranchSingle{}{\roleSource}{}{c}{
                                  \mpBranchSingle{}{\roleAlice}{}{m,n}{
                                      \mpBranchSingle{}{\roleBob}{}{o}{
                                          (\mpJustIf{o}{
                              \expUnitary[\sigma_2]{c}}) \mpSeq}}} \\ \qquad
                              (\mpJustIf{m}{
                                  \expUnitary[\sigma_1]{c}}) \mpSeq
                              (\mpJustIf{n}{
                                  \expUnitary[\sigma_2]{c}}) \mpSeq
                              \mpSel{}{\roleEnv}{}{c}{\mpNil}
                          \end{array} \\
            \mpP[S] & = & \begin{array}[t]{l}
                              \expNewQBit[a,b,c] \mpSeq
                              \expUnitary[H]{a} \mpSeq
                              \expUnitary[CNOT]{a,b} \mpSeq
                              \expUnitary[H]{a} \mpSeq
                              \expUnitary[CNOT]{b,c} \mpSeq \\\qquad
                              \mpSel{}{\roleAlice}{}{a}{
                                  \mpSel{}{\roleBob}{}{b}{
                                      \mpSel{}{\roleCharlie}{}{c}{\mpNil}
                                  }
                              }
                          \end{array}                          \\
            \mpP[E] & = & \mpSel{}{\roleAlice}{}{x}{
                \mpBranchSingle{}{\roleCharlie}{}{c}{\mpNil[c]}
            }
        \end{array}
    \]
    We sometimes use conditional $(\mpJustIf{e}{P}).Q$ as
    prefix. This can be seen as a shortcut for $\ifff
        e~\then P\subst{\mpNil}{Q}~\els Q$ or easily added to the theory.
    We also consider tuples as data types.

    We refer to~\cite{gay18} for a detailed description of the protocol, while we prefer to contrast the clarity of our notation against theirs. For instance, the process for $\roleCharlie$ in~\cite[Fig.~21]{gay18} is as follows:
    \[\gayIn{f}{c}.\gayIn{t}{m}.\gayIn{w}{n}.\gayIn{u}{o}.(\mpJustIf{o}{\expUnitary[\sigma_2]{c}}).(\mpJustIf{m}{\expUnitary[\sigma_1]{c}}).(\mpJustIf{n}{\expUnitary[\sigma_2]{c}}).\gayOutput{c}.\mpNil\]
    While there is a direct correspondence between the two terms, we argue that the use of participant names instead of channel names makes our notation easier to understand. Also, local and global types provide abstract descriptions with no correspondence in their framework.

    The typing of the participants is very close to the one of the running example, hence we do not detail it.\finex
\end{example}

\begin{example}[Quantum bit-commitment {\cite[Fig.~6]{GayN05}}]\label{ex:qbc}
    The quantum bit-commitment problem~\cite{bennett2014quantum}
    requires $\roleAlice$ to choose a $\tyBit$ value which $\roleBob$
    then attempts to guess. The key issue is that $\roleAlice$ must
    evaluate $\roleBob$’s guess with respect to her original choice of
    $\tyBit$, without changing her mind.  Similarly, $\roleBob$ must not ﬁnd out
    $\roleAlice$’s choice before making his guess.

    The code in \cite[Fig.~6]{GayN05} would roughly correspond to the
    QMPST below:
    \begin{eqnarray*}
        \trans{\roleAlice}{\roleBob}{(\tyInt)} \gtSeq \gtFmt{\mu{t}} \gtSeq \hspace{-3cm}\\
        &&(\trans{\roleAlice}{\roleBob}{(\tyQBit)} \gtSeq \trans{\roleRandom}{\roleBob}{(\tyBit)} \gtSeq \gtRecVar) +\\  &&(\trans{\roleRandom}{\roleBob}{(\tyBit)} \gtSeq \trans{\roleBob}{\roleAlice}{(\tyBit)} \gtSeq\\ && \qquad \trans{\roleAlice}{\roleBob}{(\tyBit)} \gtSeq \trans{\roleAlice}{\roleBob}{(\tyBitList)})
    \end{eqnarray*}
    The key idea is that $\roleAlice$ in order to encode a $\tyBit$ $x$ chooses a sequence of $\tyBit$s $xs$, and encodes each of them into a $\tyQBit$ using the standard basis if $x$ is $0$, and the diagonal basis otherwise. $\roleBob$ measures each $\tyQBit$ using a random basis (he gets random values from a random number generator $\roleRandom$). If the random basis matches $\roleAlice$ base, it will get the corresponding bit from $xs$, otherwise a random bit (evenly distributed). After having done this, $\roleBob$ sends to $\roleAlice$ his guess $g$, and $\roleAlice$ checks whether $\roleBob$ guessed correctly. In order to certify her answer, she sends $xs$ to $\roleBob$.

    The global type above will not capture the fact that the number of iterations of the recursive behavior depends on the integer value sent by $\roleAlice$ to $\roleBob$ in the first interaction, namely the length of $xs$. To capture this dependency one would need to introduce dependent types, which are outside the scope of the present work (we will come back to this point in the final discussion, in~\cref{sec:concl}).\

    Additionally, $\roleRandom$ does not know when the protocol should
    end ($\roleAlice$ and $\roleBob$ both know the desired number of
    iterations), indeed actually the code for $\roleRandom$
    in \cite[Fig.~6]{GayN05} is just an infinite loop. Since this may
    be a problem ($\roleRandom$ is not of the form $\mpNil[\tilde q]$ for some $\tilde q$ when the protocol
    ends, falsifying type safety (~\cref{thm:safety})), most
    MPST approaches forbid this behavior (we will come back to
    this point as well in the final discussion, in~\cref{sec:concl}).

    We can instead formalize the protocol in a different form, that does not require
    dependent types, by having $\roleAlice$ telling $\roleBob$ using
    distinct labels $\lbl{data}$ and $\lbl{stop}$ whether there are further data or not. We also
    solve the issue of $\roleRandom$ by having $\roleBob$ telling him
    when the protocol is about to finish.  This results in the global
    type in~\cref{fig:globalQBC}.
    \begin{figure*}[t]
        \centerline{
            \scalebox{\scalef}{$\displaystyle
                    \gtRec{\gtRecVar}{ \left\{
                        \begin{array}{l}
                            \gtCommSingle{\roleRandom}{\roleBob}{}{\tyBit}{} \gtSeq \gtCommBegin{\roleAlice}{\roleBob}
                            \left\{
                            \begin{array}{l}
                                \lbl{data}(\tyQBit) \gtSeq \gtCommSingle{\roleBob}{\roleRandom}{\lbl{loop}}{}{} \gtSeq \gtRecVar, \\
                                \lbl{stop} \gtSeq \gtCommSingle{\roleBob}{\roleRandom}{\lbl{stop}}{}{} \gtSeq
                                \gtCommSingle{\roleBob}{\roleAlice}{}{\tyBit}{} \gtSeq                                            \\
                                \qquad
                                \gtCommSingle{\roleAlice}{\roleBob}{}{\tyBit,\tyBitList}{} \gtSeq
                                \gtCommSingle{\roleBob}{\roleEnv}{\left\{
                                    \begin{array}{l}
                                        \lbl{verified}, \\
                                        \lbl{notVerified}
                                    \end{array}
                                    \right\}
                                }{}{}
                            \end{array}
                            \right\}
                            \\
                        \end{array}
                        \right\} }
                $}}
        \caption{Global type for Quantum Bit Commitment}\label{fig:globalQBC}
    \end{figure*}
    We can then write the corresponding system (participants are defined in~\cref{fig:qbcProcesses}):
    \[
        \sysM = \mpProc{\roleAlice}{\mpP[A]} \pp \mpProc{\roleBob}{\mpP[B]} \pp \mpProc{\roleRandom}{\mpP[R]}
    \]
    \begin{figure*}[t]
        \centerline{
            \scalebox{\scalef}{$\displaystyle
                    \begin{array}{rcl}
                        P_A & = & \mpFmt{\mathbf{def}\;X(x,xs,ys)} =                                                                        \\
                            &   & \qquad
                        \ifff xs \neq []~\then \expNewQBit[x_q]. (\mpJustIf{\hd{xs}=1}{\sigma_1(x_q)}).                                     \\
                            &   & \qquad \qquad
                        (\mpJustIf{x=1}{H(x_q)}). \mpSel{}{\roleBob}{\lbl{data}}{x_q}{\mpCall{X}{x,\tl{xs},ys}}                             \\
                            &   & \qquad \els                                                                                               
                        \mpSel{}{\roleBob}{\lbl{stop}}{}{ \mpBranch{}{\roleBob}{}{}{g}{ \mpSel{}{\roleBob}{}{(x,ys)}{\mpNil}}{}}            \\
                            &   & \mathbf{in}~\mpCall{X}{x,xs,xs}                                                                           \\

                        P_B & = & \mpFmt{\mathbf{def}\;X_V(ms,vs,a)} =                                                                      \\
                            &   & \qquad \ifff ms=[]~\then \mpSel{}{\roleEnv}{\lbl{verified}}{}{\mpNil}                                     \\
                            &   & \qquad \els \ifff \fst{\hd{ms}}\, !\!= a \lor \snd{\hd{ms}}=\hd{vs}~\then \mpCall{X_V}{\tl{ms},\tl{vs},a} \\
                            &   & \qquad \phantom{\els} \els \mpSel{}{\roleEnv}{\lbl{notVerified}}{}{\mpNil}                                \\
                            &   & \mathbf{in}~\mpFmt{\mathbf{def}\;X(ms)} =                                                                 \\
                            &   & \qquad \mpBranchRaw{}{\roleRandom}{y}.\mpBranchRaw{}{\roleAlice}{
                            \left\{
                            \begin{array}{l}
                                \lbl{data}(x). \mpSel{}{\roleRandom}{\lbl{loop}}{}{}. (\mpJustIf{y=1}{\expUnitary[H]{x})}). \measure{z}{x}. \mpCall{X}{ms@(y,z)}          \\
                                \lbl{stop}. \mpSel{}{\roleRandom}{\lbl{stop}}{}{}. \mpSel{}{\roleAlice}{}{g}{}. \mpBranchRaw{}{\roleAlice}{(a,vs)}. \mpCall{X_V}{ms,vs,a} \\
                            \end{array}
                        \right\}}                                                                                                           \\
                            &   & \mathbf{in}~\mpCall{X}{[]}                                                                                \\
                        P_R & = & \mpFmt{\mathbf{def}\;X} =
                        \expNewQBit[x_q]. \expUnitary[H]{x_q}. \measure{r}{x_q}. \mpSel{}{\roleBob}{}{r}{ \mpBranchRaw{}{\roleBob}{
                                \left\{
                                \begin{array}{l}
                                    \lbl{loop} \mpSeq \mpCall{\mpX}{} \\
                                    \lbl{stop} \mpSeq \mpNil
                                \end{array}
                                \right\}
                            }
                            ~\mathbf{in}~
                            \mpCall{\mpX}{}}
                    \end{array}
                $}}
        \caption{Processes for Quantum Bit-Commitment}\label{fig:qbcProcesses}
    \end{figure*}
    \noindent
    Using the typing rules we can derive:
    \(
        \stJudge{}{}{\sysM}{\gtG}
    \),
    provided that $x$ and $xs$ in the initial invocation of $\roleAlice$ recursive behavior are replaced by concrete values (to have a closed term).
    %
    As first step, we project the global type on the different participants, obtaining the result in~\cref{fig:qbcProjections}.
    \begin{figure*}[t]
        \centerline{
            \scalebox{\scalef}{$\displaystyle
                    \begin{array}{rcl}
                        \gtProj{\gtG}{\roleAlice}  & = & \stRec{\stRecVar}{\stIntSumSing{\roleBob}{}{\stChoice{\lbl{data}}{\tyQBit} \stSeq \stRecVar, \stChoice{\lbl{stop}}{} \stSeq \stExtSumSing{\roleBob}{}{\stChoice{}{\tyBit} \stSeq \stIntSumSing{\roleBob}{}{\stChoice{}{\tyBit,\tyBitList} \stSeq \stEnd} }}} \\
                        \gtProj{\gtG}{\roleBob}    & = & \stFmt{\mu{\stRecVar}.\stExtSumSing{\roleRandom}{}{\stChoice{}{\tyBit} \stSeq \stExtSum{\roleAlice}{}{
                                    \begin{array}{l}
                                        \stChoice{\lbl{data}}{\tyQBit} \stSeq \stIntSumSing{\roleRandom}{}{\stChoice{}{\lbl{loop}} \stSeq \stRecVar} \\
                                        \stChoice{\lbl{stop}}{} \stSeq \stIntSumSing{\roleRandom}{}{\stChoice{}{\lbl{stop}} \stSeq \stIntSumSing{\roleAlice}{}{\stChoice{}{\tyBit} \stSeq \stExtSumSing{\roleAlice}{}{\stChoice{}{\tyBit,\tyBitList} \stSeq \stIntSum{\roleEnv}{}{
                                                        \begin{array}{l}
                                                            \stChoice{\lbl{verified}}{} \stSeq \stEnd    \\
                                                            \stChoice{\lbl{notVerified}}{} \stSeq \stEnd \\
                                                        \end{array}
                                                    }    }  }}
                                    \end{array}
                        } }}                                                                                                                                                                                                                                                                                          \\
                        \gtProj{\gtG}{\roleRandom} & = & \stRec{\stRecVar}{\stIntSum{\roleBob}{}{\stChoice{}{\tyBit} \stSeq \stExtSumSing{\roleBob}{}{
                                    \begin{array}{l}
                                        \stChoice{\lbl{loop}}{} \stSeq \stRecVar \\
                                        \stChoice{\lbl{stop}}{} \stSeq \stEnd    \\
                                    \end{array}
                                }  }}
                    \end{array}
                $}}
        \caption{Projections for Quantum Bit-Commitment}\label{fig:qbcProjections}
    \end{figure*}
    For simplicity, we show the typing of $\roleRandom$ only, in~\cref{fig:qbcRandom}.\finex
    \begin{figure*}[t]
        \centerline{
            \scalebox{\scalef}{
                \inference[]{
                    \inference[]{
                        \inference[]{
                            \inference[]{
                                \inference[]{
                                    \inference[]
                                    {
                                        \stJudge{\mpX: (\emptyset, \gtProj{\gtG}{\roleRandom})}{\stEnvMap{r}{\tyBit}}{\mpCall{\mpX}{}}{\gtProj{\gtG}{\roleRandom}}
                                        \qquad
                                        \stJudge{\mpX: (\emptyset, \gtProj{\gtG}{\roleRandom})}{\stEnvMap{r}{\tyBit}}{\mpNil}{\stEnd}
                                    }
                                    {
                                        \stJudge{\mpX: (\emptyset, \gtProj{\gtG}{\roleRandom})}{\stEnvMap{r}{\tyBit}}{\mpBranchRaw{}{\roleBob}{
                                                \left\{
                                                \begin{array}{l}
                                                    \lbl{loop} \mpSeq \mpCall{\mpX}{} \\
                                                    \lbl{stop} \mpSeq \mpNil
                                                \end{array}
                                                \right\}
                                            }
                                        }
                                        {
                                            \stExtSum{\roleBob}{}{
                                                \begin{array}{l}
                                                    \stChoice{\lbl{loop}}{} \stSeq \gtProj{\gtG}{\roleRandom} \\
                                                    \stChoice{\lbl{stop}}{} \stSeq \stEnd                     \\
                                                \end{array}
                                            }
                                        }
                                    }
                                }
                                {
                                    \stJudge{\mpX: (\emptyset, \gtProj{\gtG}{\roleRandom})}{\stEnvMap{r}{\tyBit}}{\mpSel{}{\roleBob}{}{r}{ \mpBranchRaw{}{\roleBob}{
                                                \left\{
                                                \begin{array}{l}
                                                    \lbl{loop} \mpSeq \mpCall{\mpX}{} \\
                                                    \lbl{stop} \mpSeq \mpNil
                                                \end{array}
                                                \right\}
                                            }}
                                    }{\gtProj{\gtG}{\roleRandom}}
                                }
                            }
                            {
                                \qstJudge{\mpX: (\emptyset, \gtProj{\gtG}{\roleRandom})}{\emptyset}{\stEnvMap{x_q}{\tyQBit}}{\measure{r}{x_q} \mpSeq \mpSel{}{\roleBob}{}{r}{ \mpBranchRaw{}{\roleBob}{
                                            \left\{
                                            \begin{array}{l}
                                                \lbl{loop} \mpSeq \mpCall{\mpX}{} \\
                                                \lbl{stop} \mpSeq \mpNil
                                            \end{array}
                                            \right\}
                                        }}
                                }{\gtProj{\gtG}{\roleRandom}}
                            }
                        }
                        {
                            \qstJudge{\mpX: (\emptyset, \gtProj{\gtG}{\roleRandom})}{\emptyset}{\stEnvMap{x_q}{\tyQBit}}{\expUnitary[H]{x_q} \mpSeq \measure{r}{x_q} \mpSeq \mpSel{}{\roleBob}{}{r}{ \mpBranchRaw{}{\roleBob}{
                                        \left\{
                                        \begin{array}{l}
                                            \lbl{loop} \mpSeq \mpCall{\mpX}{} \\
                                            \lbl{stop} \mpSeq \mpNil
                                        \end{array}
                                        \right\}
                                    }}
                            }{\gtProj{\gtG}{\roleRandom}}
                        }
                    }{
                        \stJudge{\mpX: (\emptyset, \gtProj{\gtG}{\roleRandom})}{\emptyset}{\expNewQBit[x_q] \mpSeq \expUnitary[H]{x_q} \mpSeq \measure{r}{x_q} \mpSeq \mpSel{}{\roleBob}{}{r}{ \mpBranchRaw{}{\roleBob}{
                                    \left\{
                                    \begin{array}{l}
                                        \lbl{loop} \mpSeq \mpCall{\mpX}{} \\
                                        \lbl{stop} \mpSeq \mpNil
                                    \end{array}
                                    \right\}
                                }}
                        }{\gtProj{\gtG}{\roleRandom}}
                    } \quad \Phi
                }
                {
                    \stJudge{\emptyset}{\emptyset}
                    {
                        \mpFmt{\mathbf{def}\;X} =
                        \expNewQBit[x_q] \mpSeq \expUnitary[H]{x_q} \mpSeq \measure{r}{x_q} \mpSeq \mpSel{}{\roleBob}{}{r}{ \mpBranchRaw{}{\roleBob}{
                                \left\{
                                \begin{array}{l}
                                    \lbl{loop} \mpSeq \mpCall{\mpX}{} \\
                                    \lbl{stop} \mpSeq \mpNil
                                \end{array}
                                \right\}
                            }
                            ~\mathbf{in}~
                            \mpCall{\mpX}{}}
                    }
                    {\gtProj{\gtG}{\roleRandom}}
                }
            }}
        where $\Phi= \stJudge{\mpX: (\emptyset, \gtProj{\gtG}{\roleRandom})}{\emptyset}{\mpCall{\mpX}{}}{\gtProj{\gtG}{\roleRandom}}  $.
        \caption{Typing of $\roleRandom$ in Quantum Bit-Commitment}\label{fig:qbcRandom}
    \end{figure*}
\end{example}
\ifmain
A further case study, discussing Gottesman and Chuang's Quantum Digital
Signature~\cite{gottesman2001ds}, is available in~\cref{apx:ex:key}.
\fi
\begin{toappendix}
    \label{apx:ex:key}
    \begin{example}[Key distribution]\label{ex:key}
        We present here the protocol of key generation and distribution from
        Gottesman and Chuang Quantum Digital Signature, presented
        in~\cite{gottesman2001ds}. The description is taken from~\cite{qpzoo}\footnote{We thank M.~Delavar and S.~Singh for helping us in understanding the protocol.}. For simplicity we consider distribution of a single key, hence in the description in~\cite{qpzoo} we fix $M=1$.
        The protocol is described by the global type in~\cref{fig:kdGlobal}.

        \begin{figure*}[t]
            \centerline{
                \scalebox{\scalef}{$\displaystyle
                        \begin{array}{l}
                            \gtCommSingle{\roleSeller}{\roleBuyer}{}{\tyQBit^4}{} \gtSeq
                            \gtCommSingle{\roleSeller}{\roleVerifier}{}{\tyQBit^4}{} \gtSeq \\
                            \qquad\gtCommBegin{\roleBuyer}{\roleVerifier}
                            \left\{
                            \begin{array}{l}
                                \lbl{abort} \gtSeq \gtComm{\roleVerifier}{\roleEnv}{}{\lbl{abort}}{}{\gtEnd}, \\
                                \lbl{ok}(\tyQBit^2) \gtSeq \gtCommBegin{\roleVerifier}{\roleBuyer}
                                \left\{
                                \begin{array}{l}
                                    \lbl{abort} \gtSeq \gtComm{\roleVerifier}{\roleEnv}{}{\lbl{abort}}{}{\gtEnd}, \\
                                    \lbl{ok}(\tyQBit^2) \gtSeq \gtCommBegin{\roleBuyer}{\roleVerifier} 
                                    \left\{
                                    \begin{array}{l}
                                        \lbl{abort} \gtSeq \gtComm{\roleVerifier}{\roleEnv}{}{\lbl{abort}}{}{\gtEnd}, \\
                                        \lbl{ok} \gtSeq \gtCommBegin{\roleVerifier}{\roleEnv} \left\{
                                        \begin{array}{l}
                                            \lbl{abort} \gtSeq \gtEnd, \\
                                            \lbl{ok} \gtSeq \gtEnd
                                        \end{array}
                                        \right\}
                                    \end{array}
                                    \right\}
                                \end{array}
                                \right\}
                            \end{array}
                            \right\}
                        \end{array}
                    $}}
            \caption{Global Type of Key Distribution}\label{fig:kdGlobal}.
        \end{figure*}
        The idea of the protocol is as follows. For each message bit (say 0
        and 1), $\roleSeller$ selects a single (since we assume $M=1$) classical bit
        string randomly. This is chosen to be her private key for that
        message bit. Using this private key as input, $\roleSeller$ generates four copies of each public key
        using a quantum one-way function $\qow{\cdot}$ which returns a $\tyQBit$.
        Then $\roleSeller$ sends two copies of each public key to $\roleBuyer$ and two to $\roleVerifier$.

        In order to check that the copies of the keys are actually identical,
        the protocol relies on the quantum swap test
        from~\cite{buhrman2001quantum}, which non-destructively compares two
        $\tyQBit$s, and returns $1$ if they are equal, and either $0$ or $1$
        if they differ, where the probability of $0$ depends on the internal
        product of the two $\tyQBit$s.

        Both $\roleBuyer$ and $\roleVerifier$ test their pairs of keys, and
        abort if the test fails. Otherwise, they swap one of their keys, and
        compare their own key with the received one. If all the tests
        succeed, key distribution has been successful.

        The quantum swap test can be defined as:
        \[
            \begin{array}{rcl}
                \QS{q,q'} & = & \expNewQBit[x_a] \mpSeq \expUnitary[H]{x_a} \mpSeq \expUnitary[CSWAP]{x_a,q,q'} \mpSeq 
                \expUnitary[H]{x_a} \mpSeq \measure{x}{x_a}
            \end{array}
        \]
        where $\expUnitaryOp[H]$ is a Hadamard gate (the first occurrence allows here to generate the state $(\ket{0} + \ket{1})/\sqrt{2}$) and $\expUnitaryOp[CSWAP]$ is a controlled-swap (Fredkin gate), controlled by first qubit. We assume that the quantum swap test above returns $x$, and we use it in expressions such as $y=\QS{q_1,q_2}$.

        We can now define the system:
        \[\sysM = \named{\roleSeller}{P_S} \pp \named{\roleBuyer}{P_B} \pp \named{\roleVerifier}{P_V}\]
        where the definition of processes is in~\cref{fig:kdProcesses}.
        \begin{figure*}
            \centerline{
                \scalebox{\scalef}{$\displaystyle
                        \begin{array}{rcl}
                            P_S & = & \mpSel{}{\roleBuyer}{}{\qow{k_0},\qow{k_0},\qow{k_1},\qow{k_1}}{ \mpSel{}{\roleVerifier}{}{\qow{k_0},\qow{k_0},\qow{k_1},\qow{k_1}}{\mpNil}} \\
                            P_B & = & \mpBranchSingle{}{\roleSeller}{}{qk_0^1,qk_0^2,qk_1^1,qk_1^2}{} \mpSeq r_0=\QS{qk_0^1,qk_0^2} \mpSeq r_1=\QS{qk_1^1,qk_1^2} \mpSeq           \\&&\qquad\mpIfMultZ{\neg r_0 \lor \neg r_1}{\mpSel{}{\roleVerifier}{\lbl{abort}}{}{\mpNil[qk_0^1,qk_0^2,qk_1^1,qk_1^2]}}{
                                \mpSel{}{\roleVerifier}{\lbl{ok}}{qk_0^1,qk_1^1}{} \mpSeq \mpBranchRaw{}{\roleVerifier}{\left\{
                                    \begin{array}{l}
                                        \lbl{abort} \mpSeq \mpNil[qk_0^2,qk_1^2],                         \\
                                        \lbl{ok}(qk_0^v,qk_1^v) \mpSeq
                                        r_{0f}=\QS{qk_0^v,qk_0^2} \mpSeq r_{1f}=\QS{qk_1^v,qk_1^2} \mpSeq \\ \qquad \mpIfMultZ{\neg r_{0f} \lor \neg r_{1f}}{\mpSel{}{\roleVerifier}{\lbl{abort}}{}{\mpNil[qk_0^2,qk_1^2,qk_0^v,qk_1^v]}}{\mpSel{}{\roleVerifier}{\lbl{ok}}{}{\mpNil[qk_0^2,qk_1^2q,k_0^v,qk_1^v]}}
                                    \end{array}
                                    \right\}}
                            }                                                                                                                                                      \\
                            P_V & = & \mpBranchSingle{}{\roleSeller}{}{qk_0^3,qk_0^4,qk_1^3,qk_1^4}{} \mpSeq r_0=\QS{qk_0^3,qk_0^4} \mpSeq r_1=\QS{qk_1^3,qk_1^4}.                 \\&&\qquad\mpBranchRaw{}{\roleBuyer}{
                                \left\{
                                \begin{array}{l}
                                    \lbl{abort} \mpSeq \mpSel{}{\roleEnv}{\lbl{abort}}{}{\mpNil[qk_0^3,qk_0^4,qk_1^3,qk_1^4]} \\
                                    \lbl{ok}(qk_0^b,qk_1^b) \mpSeq
                                    \mpIfMultZ{\neg r_0 \lor \neg r_1}
                                    {\mpSeq \mpSel{}{\roleBuyer}{\lbl{abort}}{}{\mpSeq \mpSel{}{\roleEnv}{\lbl{ok}}{}{\mpNil[qk_0^3,qk_0^4,qk_1^3,qk_1^4,qk_0^b,qk_1^b]}}}
                                    {
                                        \mpSel{}{\roleBuyer}{\lbl{ok}}{qk_0^3,qk_1^3}{}{} \mpSeq
                                        \mpBranchRaw{}{\roleBuyer}{
                                            \left\{
                                            \begin{array}{l}
                                                \lbl{abort} \mpSeq \mpSel{}{\roleEnv}{\lbl{abort}}{}{\mpNil[qk_0^4,qk_1^4,qk_0^b,qk_1^b]}, \\
                                                \mpSeq r_{0g}=\QS{qk_0^b,qk_0^4} \mpSeq r_{1g}=\QS{qk_1^b,qk_1^4} \mpSeq                   \\ \qquad
                                                \mpIfMultZ{\neg r_{0g} \lor \neg r_{1g}}
                                                {\mpSeq \mpSel{}{\roleEnv}{\lbl{abort}}{}{\mpNil[qk_0^4,qk_1^4,qk_0^b,qk_1^b]}}
                                                {\mpSeq \mpSel{}{\roleEnv}{\lbl{ok}}{}{\mpNil[qk_0^4,qk_1^4,qk_0^b,qk_1^b]}}
                                            \end{array}
                                            \right\}
                                        }

                                    }
                                \end{array}
                                \right\}
                            }                                                                                                                                                      \\
                            P_E & = & \mpBranchSingle{}{\roleVerifier}{
                                \left\{
                                \begin{array}{l}
                                    \lbl{abort} \mpSeq \mpNil, \\
                                    \lbl{ok} \mpSeq \mpNil
                                \end{array}
                                \right\}
                            }{}{}
                        \end{array}$
                }}
            where $k_0$ and $k_1$ are vectors of $m$ bits.
            \caption{Processes of Key Distribution}\label{fig:kdProcesses}
        \end{figure*}
        Function $\qowOp$ performs quantum one-way map on the argument, which takes $n$ bits and computes a qubit. We need an extension of the syntax to allow for functions which take bits and return qubits. Array notation is for element selection.

        The projections are in~\cref{fig:kdProjections}.

        \begin{figure*}[t]
            \centerline{
                \scalebox{\scalef}{$\displaystyle
                        \begin{array}{rcl}
                            \gtProj{\gtG}{\roleSeller}   & = & \stIntSumSing{\roleBuyer}{}{\stChoice{}{\tyQBit^4} \stSeq \stIntSumSing{\roleVerifier}{}{\stChoice{}{\tyQBit^4} \stSeq \stEnd}} \\
                            \gtProj{\gtG}{\roleBuyer}    & = & \stExtSumSing{\roleSeller}{}{\stChoice{}{\tyQBit^4} \stSeq \stIntSum{\roleVerifier}{}{
                                    \begin{array}{l}
                                        \stChoice{\lbl{abort}}{} \stSeq \stEnd, \\
                                        \stChoice{\lbl{ok}}{\tyQBit^2} \stSeq \stExtSum{\roleVerifier}{}{
                                            \begin{array}{l}
                                                \stChoice{\lbl{abort}}{} \stSeq \stEnd, \\
                                                \stChoice{\lbl{ok}}{\tyQBit^2} \stSeq  \stIntSum{\roleVerifier}{}{
                                                    \begin{array}{l}
                                                        \stChoice{\lbl{abort}}{} \stSeq \stEnd, \\
                                                        \stChoice{\lbl{ok}}{\tyQBit^2} \stSeq \stEnd
                                                    \end{array}}
                                            \end{array}
                                        }
                                    \end{array}
                            }}                                                                                                                                                                 \\
                            \gtProj{\gtG}{\roleVerifier} & = & \stExtSumSing{\roleSeller}{}{\stChoice{}{\tyQBit^4} \stSeq \stExtSum{\roleBuyer}{}{
                                    \begin{array}{l}
                                        \stChoice{\lbl{abort}}{} \stSeq \stEnd, \\
                                        \stChoice{\lbl{ok}}{\tyQBit^2} \stSeq \stIntSum{\roleBuyer}{}{
                                            \begin{array}{l}
                                                \stChoice{\lbl{abort}}{} \stSeq \stEnd, \\
                                                \stChoice{\lbl{ok}}{\tyQBit^2} \stSeq  \stExtSum{\roleBuyer}{}{
                                                    \begin{array}{l}
                                                        \stChoice{\lbl{abort}}{} \stSeq \stEnd, \\
                                                        \stChoice{\lbl{ok}}{\tyQBit^2} \stSeq \stIntSum{\roleEnv}{}{
                                                            \begin{array}{l}
                                                                \stChoice{\lbl{abort}}{} \stSeq \stEnd, \\
                                                                \stChoice{\lbl{ok}}{} \stSeq \stEnd
                                                            \end{array}
                                                        }
                                                    \end{array}}
                                            \end{array}
                                        }
                                    \end{array}
                                }}
                        \end{array}
                    $}}
            \caption{Projections of Key Distribution}\label{fig:kdProjections}
        \end{figure*}

        From the informal protocol description above it may seem that
        $\roleBuyer$ and $\roleVerifier$ can exchange the results
        of their tests in any order. This is however considered harmful
        since if the two participants do not agree on the order, then they may
        decide to both start receiving, resulting in a deadlock. Such a
        protocol would not be typable using MPSTs (it violates type
        safety (~\cref{thm:safety})), indeed in our formalization we had to
        force an ordering. A protocol forcing the opposite order would be
        typable as well. Recent approaches such as~\cite{CasalMV22} show that
        the protocol with no pre-defined order can be managed for synchronous
        communications (like in our setting) using mixed choice, which are
        however not supported by the current typing rules. This is left for
        future work.\finex
    \end{example}

\end{toappendix}

\section{Discussion}
\label{sec:concl}
We have extended MPSTs so to provide a clean formalism to describe and reason about quantum protocols.
There are other languages for describing quantum protocols,
such as imperative languages like LanQ~\cite{LanQ} and QMCLANG~\cite{Papanikolaou09,DavidsonGMNP12},
and process calculi such as CQP~\cite{GayN05} and $\textrm{CCS}^\textrm{q}$~\cite{gay18}.
\emph{All} of them rely on named channels for communication, and channels are typed
with the type of their payload (and linear in some approaches),
hence multiple channels are needed between the same participants.
We believe our approach of specifying in send (resp.~receive) operations the target (resp.~sender) of the
message makes for a clearer description. All the approaches above are also based on classical types, not behavioral
types, hence do not provide a global type as abstract description of the protocol.

\paragraph*{Quantum Session Types}
We purposefully based our approach on a minimal MPST system
(taking inspiration from~\cite{GlabbeekHH21,gentleIntro}), to
highlight the interplay between MPSTs and quantum computing,
and to understand which advanced features of session types are useful
to type protocols in the area.
While we have found no need for multiple sessions or delegation,
\cref{ex:qbc} could be modeled in a way closer to the specification
in~\cite[Fig.~6]{GayN05} by relying on value-dependent session types
(see~\cite{toninho2011dependent}). Indeed,
one would like to specify that the number of iterations depends on an
integer exchanged beforehand.

The same example has also shown that sometimes requiring all the
processes to complete their program is too restrictive. In~\cref{ex:qbc},
$\roleRandom$ could be more easily modeled as an
infinite loop generating random numbers, and the fact that it could
generate further random numbers after the end of the protocol is not
an issue from the practical point of view. However, to type such a
formalisation of the protocol one would need, beyond changing the
typing rules, to weaken type safety (\cref{thm:safety}). One could
specify that some participants do not need to terminate their
behaviour. Such a kind of approach has been studied
in~\cite{BD23}. Note that selective participation (as discussed
in~\cite{HuY17,HFDG21,GheriLSTY22}), allowing a participant to join
only in some branches of a protocol, would not be helpful here, since
indeed $\roleRandom$ is needed in all the executions.

In the future we want to tackle the two limitations discussed above
(as well as mixed choice, cf.~\cref{ex:key}).  Also, it would be
interesting to consider refinement types. While not needed from the
modeling point of view, they would allow to prove stronger properties,
e.g., about possible results of measure operations. This would need
however to combine refinement types with probabilities (the result of
a measurement operation is probabilistic) and to capture the
phenomenon of entanglement, which relates the behavior of different
qubits.

\paragraph*{Quantum Processes}
The particular abstraction of quantum processes we have chosen is based on CQP~\cite{GayN05},
where processes can perform quantum operations on quantum data,
and exchange messages on \emph{channels} with other processes --
these channels allow for classical communication with \emph{quantum} references,
which are pointers to qubits in a \emph{global} quantum register.
Ergo, these are not \emph{quantum channels} in the sense of quantum information theory
(see~\cite[Ch.12]{nielsen2010quantum}),
because processes do not have \emph{local} quantum state --
however, having a global quantum state with local quantum references appears to be expressive enough --
whether this is a limitation or a feature of our approach remains to be investigated.

There exist rigorous mathematical frameworks for quantum networks~\cite{giulio2009theoretical}
and quantum protocols~\cite{abramskyCategoricalSemanticsQuantum2004} (to name a few).
We do not know any formal calculi for quantum processes and quantum channels based on these frameworks,
and it is not yet clear how our approach of extending MPSTs to quantum computing could be related to them.
Session types are a completely syntactic approach,
but there are recent attempts to give them denotational semantics
in terms of event structures~\cite{CASTELLANI2023100844},
or using category
theory~\cite{atkeyObservedCommunicationSemantics2017,kavanaghDomainSemanticsHigherOrder2020,choudhuryClassicalProcessesModern2023}.
In fact, the categorical structure of session types is the same as the categorical structure of quantum processes
(compact closure or star-autonomy) -- they both require a notion of strong duality.
This suggests, on the one hand, a way to use session types to build a type theory for quantum processes,
and on the other hand, a way to verify correctness of typed quantum processes using denotational semantics,
which we hope to explore in future work.

\renewcommand{\appendixsectionformat}[2]{
  {Supplementary material for Section~#1 (#2)}
}

\ifmain\clearpage\fi
\printbibliography
\appendix

\end{document}